%% file: main_paper.tex
\documentclass[12pt]{amsart}

	\input{macros}

	\makeatletter
	\newcommand{\customlabel}[2]{%
	\protected@write \@auxout {}{\string \newlabel {#1}{{#2}{}}}}
	\makeatother
	\customlabel{appsec:Estimation and closed preference sets}{S.1}

\begin{document}
\title[Recovering Preferences from Finite Data]{Recovering Preferences from Finite Data}

\author[Chambers]{Christopher P. Chambers}
\author[Echenique]{Federico Echenique}
\author[Lambert]{Nicolas S. Lambert}

\address[Chambers]{Department of Economics, Georgetown University}
\address[Echenique]{Division of the Humanities and Social Sciences,
  California Institute of Technology}
\address[Lambert]{Department of Economics, Massachusetts Institute of Technology}

\begin{abstract}
	We study preferences estimated from finite choice experiments and provide sufficient conditions for convergence to a unique underlying ``true'' preference. Our conditions are weak, and therefore valid in a wide range of economic environments. We develop applications to expected utility theory, choice over consumption bundles, menu choice and intertemporal consumption. Our framework unifies the revealed preference tradition with models that allow for errors.
\end{abstract}

\thanks{An early version circulated under the title  ``Preference Identification.'' 
We are particularly grateful to Jeroen Swinkels, the coeditor and six anonymous referees for insightful comments and suggestions. We are thankful to Pathikrit Basu for finding a mistake in a prior draft of the paper and for valuable suggestions. We thank seminar participants at many institutions, the audiences of conferences and workshops at HEC Paris, Paris School of Economics, UC Berkeley, University of Pennsylvania, University of Warwick, University of York, Oxford University, and Virginia Tech. Echenique thanks the National Science Foundation for financial support (Grants SES-1558757 and CNS-518941), and the Simons Institute at UC Berkeley for its hospitality. Lambert thanks Microsoft Research and the Cowles Foundation at Yale University for their hospitality and financial support.}
\date{First version: April 10, 2017; This version: \today{}.}

\maketitle

\section{Introduction}\label{sec:Introduction}


This paper provides conditions that guarantee asymptotic,  large-sample, nonparametric recoverability of preferences from binary choice data. 
We imagine an experimenter, Alice, offering a sequence of binary choice problems to a subject, Bob. For each choice problem, Bob is presented with a pair of alternatives and is asked to choose one (for example, alternatives could be lotteries over a collection of prizes). Alice wants to ensure that, if she observes Bob on sufficiently many choice problems, then she can approximate his preference over the entire set of alternatives to an arbitrary degree of precision. We study when this approximation possible. The goal is to have conditions that are easy to check and apply broadly.

Our approach is two-pronged. 
Our first model is anchored in the classical revealed preference tradition, whereby Alice seeks to rationalize the choice data exactly. The model is deterministic. Alice designs a fixed experiment, and hypothesizes that Bob chooses perfectly in accordance with his preference.
Our second model is statistical. The selection of experiments is random, and Alice supposes that Bob's choices are either observed with some error, or made with error.
In both models, we provide conditions for the underlying data-generating preference to be recovered in the limit as the available data grows large.
These conditions concern both the experimental design and the preference environment being considered. 

The main substantive condition is the \emph{local strictness} of preferences, a property first described by \citet{border1994dynamic}. Local strictness generalizes the familiar notion of local nonsatiation. A locally strict preference means that whenever $x$ is at least as good as $y$, there are alternatives $x'$ and $y'$, near $x$ and $y$ respectively, and for which $x'$ is strictly better than $y'$.
We prove that, together with technical conditions, local strictness ensures the convergence of any sequence of rationalizing preferences to the unique underlying preference governing the subject's choices as the number of observations goes to infinity.
In the statistical model, we introduce an estimator based on minimizing the Kemeny distance to the observed choices. Again imposing local strictness, we prove that this estimator is consistent, and provide general convergence rates.

The usefulness of our results is illustrated with applications to expected-utility theory and other environments where preferences are defined from utility functions; to monotone preferences in consumption theory, preferences over menus, and exponential discounting in intertemporal choice. 
In all these applications, we show how large finite experiments can approximate a preference of the appropriate kind. 

Our approach works with the choice-theoretic notion of partial observability, as in \citet{afriat}: when Bob is presented with a choice, he must select no more than one alternative. Therefore, choosing one alternative over another does not preclude the possibility that Bob would have been equally happy with the other alternative.
In Afriat's case, partial observability has wide-ranging implications, famously rendering  the concavity of utility nontestable.\footnote{\citet*{CES} provide a discussion of what partial observability entails.}  Partial observability generally leads to partial identification and hinders estimation.  
In our framework, preferences can still be fully learned in spite of their partial observability.

We allow for very general sets of preferences, which translate into a ``model-free'' approach. If, say, Alice is interested in exponential discounting, then she can estimate a preference without the need to impose this assumption on the data. If Bob is indeed discounting exponentially, then the preference estimates are guaranteed to converge to a preference that follows exponential discounting. And if the preference estimates do not converge to such a preference, then Alice may conclude that the exponential discounting hypothesis is incorrect. She can then evaluate the degree to which Bob's preference diverges from exponential discounting. 
The model-free aspect is also present in the statistical model: Alice can be agnostic about how the alternatives presented to subjects are sampled, and how subjects are assumed to make mistakes. Specifically, the estimator does not require knowledge of sampling or error probabilities.
 
Overall, our framework combines the elements of different traditions in economic modeling: the nonparametric approach and the finite amount of data in revealed preference analysis, the pairwise comparisons in decision theory and experiments (both online and laboratory), and the source of random errors in empirical research and econometrics.

The paper proceeds as follows. 
The remainder of this section reviews related works. Section~\ref{sec:Model} describes the model. 
Section~\ref{sec:Main results} provides the main results. Sections~\ref{sec:Preferences from utilities} and~\ref{sec:Application to monotone preferences} put these results to work in several economic environments. Section~\ref{sec:Concluding discussion} concludes with a discussion. Proofs are relegated to the appendices.

\subsection{Related literature.}\label{subsec:Related literature}

The literature on revealed-preference theory has been primarily concerned with the question of whether observed behavior conforms with standard models in economic theory. The workhorse of this literature, Afriat's theorem \citep{afriat,varian1982nonparametric}, works with the classical model of consumer demand with linear budgets and finite data, and has been expanded in many directions.\footnote{
For example, 
\citet{matzkin1991axioms} and \citet{forges2009} work with nonlinear budget sets, \citet*{chavas1993generalized} and \citet*{nishimura} work with general choice problems. Some extensions were also developed for multiperson equilibrium models, as in \citet*{brown1996testable}.
}
This line of research focuses for the most part on constructing revealed preference tests, and discussions of preference recoverability---as, for instance, in \citealp*{varian1982nonparametric}, or more recently \citealp*{cherchye2011revealed}---deal with bounding the sets of possible rationalizing preferences.

Closer to our work, in the context of consumer demand with linear budgets, \citet{mascolell78} introduces an ``income-Lipschitz'' condition and shows that, under this and a boundary condition, any sequence of preferences that rationalizes a sufficiently rich sequence of observations converges to the unique preference that rationalizes the entire demand function. \citet{forges2009} derive the analog of \citeauthor{mascolell78}'s results for nonlinear budget sets. 
In a model of dynamic asset markets, \citet{kubler2015identification} derive conditions that permit the identification of utilities and beliefs of a subjective expected utility maximizer, and show that preference estimates from finite data converge to the unique underlying preference as the number of observations grows large.  \citet*{polemarchakis2016identification}  give conditions under which the identification of preferences is possible, and demonstrate the convergence of preferences estimated from finite data to the unique underlying preference. More recently, \citet*{polemarchakiskubler2020} consider identification of preferences from finite data, with emphasis on applications to demand and aggregate demand.
As in the above works, our paper concerns the convergence of preferences that rationalize finite but unbounded data to the underlying data-generating preference, following large sample theory.\label{ref:R3LST1} We differ in three ways: first, we focus on data from pairwise choices instead of choice from budgets; second, we work with general environments beyond choice over commodities; third, we provide a broad sufficient condition on the class of preferences under consideration. 
Because we abstract away from specific economic environments, our results are applicable across different domains, such as choice of menus, under uncertainty, of intertemporal streams, lotteries, or consumption bundles.

Experimentalists and decision theorists also have an obvious interest in preference recovery from pairwise choices, but little is known about the behavior of preference estimates from finite data. Decision theory works often include a discussion of identification, but presume access to the agent's full preference relation. In contrast, we are interested to know if and when the preference relation can be inferred from the data. 
In demand theory, there are many studies devoted to the problem of identification---known as the integrability problem---assuming access to a demand function defined on all prices. 
\citet{matzkin2006identification} considers economy-wide data, and uses equilibrium as a means to identify consumers' utilities. 
In recent work, \citet{gorno} studies the general problem of identification under partial observability. Gorno provides conditions on decision problems and sets of admissible preferences to ensure identification. Our research diverges from his, and from other studies of identification and integrability, in that we focus on large-sample estimation from finite data.

Another stream of literature combines nonparametric econometric methods with revealed preference theory. In demand analysis, \citet*{blundell2008best} design a statistical test for the revealed preference conditions to be satisfied. Observing that demand responses to price changes can be  represented by a set of moment inequalities, they appeal to results on moment inequality estimators by \citet*{manski2003partial}, \citet*{chernozhukovhongtamer} and \citet*{andrews2009validity}. 
However, these results on partial identification do not apply to the general environments we consider. Instead, we work with the classical large-sample theory\label{ref:R3LST2} for $M$-estimators \citep{amemiya1985advanced,newey1994large} but derive conditions for consistency without making additional compactness and equicontinuity assumptions (these assumptions are particularly substantial in a nonparametric estimation problem like ours, see Section~\ref{subsec:On the connection with econometrics} for a detailed discussion).\label{R2:Halevypage}
\citet*{halevy2018parametric} develop a method for estimating parametric models by minimizing the incompatibility of choice behavior with the proposed model, in the same spirit as our Kemeny-distance estimator. In contrast to our work, their methodology assumes data on choices from linear budgets, and adopts a money-metric version of Afriat's and Varian's ``critical cost efficiency index'' as a measure of distance. A crucial component of their analysis is to decompose measures in loss due to parametric misspecification, and loss due to inconsistency with rationality. More closely related to our paper, \citet{matzkin2003nonparametric} and  \citet*{blundell2010stochastic} consider identification in an econometric model of stochastic demand data (see \citealp{matzkin2007heterogeneous}, for a general discussion).
Recently, \citet*{basu2018learnability} investigate the learnability of four standard models of choice under uncertainty using the notion of Probably Approximately Correct (PAC) learning from computational learning theory.\label{R2:PACpage} \citet*{basu2019learnability} applies several other measures of model complexity to a study of stochastic choice.


Finally, a literature in political science (\citealp*{poolerosenthal1985}, \citealp*{jackman_2001}, \citealp*{clinton_jackman_rivers_2004}, are seminal) focuses on binary choice data (roll-call votes), but considers specific parametric models of spatial voting, and uses Bayesian methods for the most part. Our results are broadly applicable to the same data as in this literature.

\section{Model}\label{sec:Model}

The model features an experimenter, Alice, and a subject, Bob.  Bob has a preference over a set of alternatives $X$, which is a topological space. Alice would like to learn Bob's preference through the device of a choice experiment.\footnote{Such experiments, done on a large scale, with large sample size, include for example \citet*{vonGaudecker2011}, \citet*{chapman2018econographics} and \citet*{falkQJE2018}. Alternatively, we may think of Alice as a researcher, and Bob an individual she has observed in the field. For example, Bob could be a congressman who votes among pairs of competing bills \citep{poolerosenthal1985}.}

By \df{preference relation} or simply \df{preference} we mean a binary relation $\succeq$ over $X$ that is continuous and complete.\footnote{Completeness means that for all pairs of alternatives $(x,y)$, $x \succeq y$ or $y \succeq x$. Continuity means that ${\succeq}$ as a subset of the product space $X \times X$ is closed; more intuitively, if $x$ is not preferred to $y$, then $x'$ is also not preferred to $y'$ for all pairs $(x',y')$ in the vicinity of $(x,y)$. Completeness is standard and continuity is a necessary regularity condition, without it, no meaningful inferences can be made with any finite amount of data.} In formal terms, $\succeq$ is the set of pairs $(x,y) \in X \times X$ such that $x$ is at least weakly preferred to $y$, denoted by $x \succeq y$.
Associated to any given preference ${\succeq}$ are its asymmetric part ${\succ}$ (strict preference) and its symmetric part ${\sim}$ (indifference); that is, $x \succ y$ means that $x \succeq y$ but $y \nsucceq x$, while $x \sim y$ indicates that both $x \succeq y$ and $y \succeq x$. We do not require that preferences be transitive.

If Alice's goal is to infer Bob's preference from his behavior, it is clear that she must somehow  discipline the set of preferences being considered.  With partial observability, it is very easy to find a preference that rationalizes empirical data. For example,  complete indifference rationalizes any observed behavior. 
Throughout the paper, $\P$ denotes the class of preferences being considered; we think of $\P$ as a set that encompasses the possible preferences the subject may have. We refer to a pair $(X,\P)$ as a \df{preference environment}.

Alice collects information about Bob through a finite experiment, in which Bob confronts a fixed number of \df{binary choice problems}. In each binary choice problem, Bob is presented with an unordered pair of alternatives, and is asked to choose exactly one of the two alternatives. An \df{experiment of size $n$} is represented by a collection $\Sigma_n = \{B_1, \dots, B_n\}$, where $B_k = \{x_k,y_k\}$ is an unordered pair of alternatives that captures a binary choice problem. Note that binary choices are used for simplicity: one could have more than two choices instead.

To study the large-sample properties of estimated preferences,\label{ref:R3LST3} we consider not just one experiment, but a set of growing experiments indexed by their size, of the form $\{\Sigma_1, \Sigma_2, \dots\}$, where $\Sigma_n$ is an experiment of size $n$ and $\Sigma_{n} \subset \Sigma_{n+1}$. In the sequel, $\Sigma_n$ always denotes an experiment of size $n$, and the inclusion property $\Sigma_{n} \subset \Sigma_{n+1}$ is implicitly assumed. We use the abbreviated notation $\{\Sigma_n\}$ to denote a set of (growing) experiments.

The behavior of a subject who decides over binary choice problems is encoded in a single-valued choice function $c$ that maps unordered pairs of alternatives to alternatives. It records, for every possible binary choice problem $\{x,y\} \subset X$, the alternative $c(\{x,y\}) \in \{x,y\}$ that is chosen. We refer to $c$ as the \df{choice function}, and impose no a priori restrictions on choice functions. 

We follow two traditions in economic modeling. 
The first tradition is classical revealed preference theory, in which the choice problems of an experiment are selected arbitrarily, and the experimenter seeks to exactly rationalize observed behavior, as in the classical works of \citet{afriat}, \citet{mascolell78}, and \citet{varian1982nonparametric}. In this theory, Bob is assumed to possess a preference and to make choices that comply perfectly with this preference. Alice looks for a preference that fits exactly the empirical observations. However, this theory does not account for errors, while empirical work often tries to accommodate errors. 

The second tradition tackles this problem by imposing a statistical model on the subject's choices. 
The subject is presented with choices drawn at random, either because the experimental design is explicitly random (as, for example, in \citealp*{ahn2014}, \citealp*{choi2014more}, \citealp*{carvalho2016poverty}, or \citealp*{carvalho2017complexity}), or because the experimenter uses observational data in which she has no control over the problems the subject faces. In the statistical tradition, Alice continues to assume that Bob has an underlying preference, but she allows for his behavior to deviate from what his preference dictates. Alice looks for a preference that fits the observed behavior the best.\footnote{See also \citet*{grant2016theory} for a general study of experimental designs that are tolerant to small deviations in the subject's perception of the experiments.}

\subsection{Revealed preference models.}\label{subsec:Revealed preference models}

In a revealed preference model, experiments are designed arbitrarily by the experimenter.
The primitives are the preference environment $(X,\P)$, and the set of experiments $\{\Sigma_n\}$. We refer to this model by the triple $(X,\P,\{\Sigma_n\})$. 

Recall that when presented with a pair of alternatives $\{x,y\}$, Bob is asked to choose \emph{between} $x$ and $y$---he cannot choose both.  In the language of \citet*{CES}, our model of choice features partial observability, as in the original work of \citet{afriat}.\footnote{The tradition in revealed preference theory (and in studies of integrability) prior to Afriat was to exactly rationalize a demand function. In Afriat's model, the observed choices are contained in the rationalizing demand, and in consequence concavity of utility is not testable. See \citet*{CES} for a detailed discussion and exploration of the consequences of partial observability.}
With partial observability, the appropriate concept of rationalization is weak rationalization:
Given a choice function $c$ describing the subject's behavior, and given an experiment $\Sigma_n$, we say that a preference ${\succeq}$ \df{weakly rationalizes} the observed choices on $\Sigma_n$, or simply \df{rationalizes} the observed choices on $\Sigma_n$, if the experiment outcomes are compatible with the subject's preference: for every $\{x, y\} \in \Sigma_n$, $c(\{x, y\}) \succeq x$ and $c(\{x, y\}) \succeq y$. 
Similarly, we say that ${\succeq}$ \df{rationalizes} the choice function $c$ if, for every $x, y \in X$, $c(\{x, y\}) \succeq x$ and $c(\{x, y\}) \succeq y$.
Hence, weak rationalization does not allow for Bob to choose in contradiction with his preference, but allows for Bob not to reveal the totality of what his preference implies.

\subsection{Statistical preference models.}\label{subsec:Statistical preference models}

In a statistical preference model, experiments are composed of randomly-selected choice problems. More precisely, the choice problems $B_1=\{x_1,y_1\}, \dots, B_n=\{x_n,y_n\}$ that make up the experiment $\Sigma_n$ are generated by drawing the alternatives $x_k, y_k$ in each $B_k$ at random from $X$, independently and identically according to some probability measure $\lambda$ ($X$ is endowed with the usual Borel $\sigma$-algebra). We also use $\la$ to denote the product measure on $X\times X$.

Bob's behavior is guided by his preference, but in every choice problem where Bob is not indifferent, he may make a mistake by choosing an alternative that is \emph{not} preferred.
The corresponding choice function is therefore random. It is determined by an \df{error probability function} $q :\P\times X\times X\rightarrow [0,1]$ that quantifies the extent to which a subject is prone to making errors.  Assume that, for all ${\succeq} \in \P$, $q({\succeq}; \cdot,\cdot)$ is measurable on $X \times X$.

When a subject with preference ${\succeq}$ confronts the binary choice problem $\{x,y\}$, he chooses $x$ over $y$ with probability $q({\succeq}; x, y)$, and chooses $y$ over $x$ with the complementary probability. 
We assume that if $x \succ y$ then $x$ is more likely to be chosen, that is, $q({\succeq};x,y) > 1/2$. When $x \sim y$ we make no particular assumption.\footnote{Strictly speaking, an experiment $\Sigma_n$ is now a multiset, and Bob could in principle face the same choice problem more than once. The model would then have to take a stand on whether Bob's choices are constant over such repetitions, or  give rise to independent choice draws. Our assumptions, however, guarantee that repetitions occur with probability zero. So these modeling assumptions are irrelevant.\label{R3:repetition}}


The primitives of a statistical preference model are the preference environment given by $X$ and $\P$, the probability measure $\la$ according to which alternatives are drawn, and the error probability function $q$. We refer to this model by the tuple $(X, \P, \la, q)$.

\section{Main results}\label{sec:Main results}

This section provides general results on the convergence of preferences. Throughout, we use the following notion of convergence: a sequence of preferences $\{ {\succeq}_n\}_{n \in \Na}$ converges to a preference ${\succeq}^*$, written ${\succeq}_n \rightarrow {\succeq}^*$ for short, when the following two conditions are satisfied:
\begin{enumerate}
	\item For all alternatives $x^*,y^*$ with $x^* \succeq^* y^*$, there exists a sequence of pairs of alternatives $\{ (x_{n},y_{n}) \}_{n \in \Na}$ converging to $(x^*,y^*)$ such that $x_{n} \succeq_{n} y_{n}$ for all $n \in \Na$.
	\item For all subsequences $\{ {\succeq}_{n_k} \}_{k \in \Na}$, and all pairs of alternatives $(x^*,y^*)$ that are the limit of a sequence $\{ (x_{n_k},y_{n_k}) \}_{k \in \Na}$ satisfying $x_{n_k} \succeq_{n_k} y_{n_k}$ for all $k$, we have $x^* \succeq^* y^*$.
\end{enumerate}
Under the assumptions we shall impose, these conditions define convergence in the \df{closed convergence topology}. Throughout, we endow the space of preferences and binary relations  with this topology. The closed convergence topology is a common topology for spaces of sets, such as binary relations, and is the standard topology used for spaces of preferences \citep{kannai1970continuity,HILDENBRAND1970161}. It is particularly well suited to the concept of partial observability; we discuss the choice of topology in Section~\ref{subsec:On the convergence of preferences}. 

Under conditions that are satisfied in our model, the closed convergence topology on the space of preferences is metrizable, making it possible to quantify approximations and speak of convergence rates. \emph{We fix, and denote by $\rho$, one compatible metric.}
In particular, if $X$ is compact and metrizable, then we can choose as $\rho$ the usual Hausdorff metric for the product space $X \times X$. The notion of closed convergence then coincides with the notion of Hausdorff convergence.\footnote{Our results allow for $X$ to be only locally compact. In this case, $\rho$ may still be chosen to coincide with the Hausdorff metric on subsets of the product space $X_{\infty} \times X_{\infty}$, where $X_{\infty}$ is the one-point compactification of $X$ together with some metric generating $X_{\infty}$. See \citet{aliprantis2006infinite} for details.}
We connect the topology on preferences with familiar topologies on spaces of utility functions in Section~\ref{sec:Preferences from utilities}, and, for the case of parameterized classes of preferences, with metrics on the parameter space in Section~\ref{lem:compatible-metric}.

\subsection{Convergence in revealed preference models.}\label{subsec:Convergence in revealed preference models}

Our first main result states that convergence of rationalizing preference obtains under certain assumptions on the model primitives. Given a revealed preference model $(X,\P,\{\Sigma_n\})$, consider the following assumptions.

\newlist{condrevealed}{enumerate}{1}
\setlist[condrevealed]{label=\scshape Assumption \arabic*.~~, ref=(\arabic*), wide=25pt,widest=99,leftmargin=*}

\medskip
\begin{condrevealed}
	\item $X$ is a locally compact, separable, and completely metrizable space.\label{cond:1}
\end{condrevealed}
\medskip

Assumption~\ref{cond:1} puts a necessary structure on the set of alternatives. It is satisfied in many common economic environments, as we show in Sections~\ref{sec:Preferences from utilities} and~\ref{sec:Application to monotone preferences}.

The next assumption disciplines the class of the preferences being considered. The central property that allows for  meaningful preference recovery is local strictness. This property rules out ``thick'' indifference curves, in the spirit of the local nonsatiation property of consumer theory. A preference $\succeq$ is \df{locally strict} if for every $x, y \in X$ with $x \succeq y$, and every neighborhood $V$ of $(x,y)$ in $X\times X$, there exists $(x',y')\in V$ with $x' \succ y'$ \citep{border1994dynamic}.

\medskip
\begin{condrevealed}[resume]
	\item $\P$ is a closed set of locally strict preferences. \label{cond:2}
\end{condrevealed}
\medskip

The requirement that $\P$ be closed  is essential. 
For example, suppose that $X=[0,1]$, that $\P$ is the set of all locally strict preferences, and that Bob prefers larger numbers, so that $x \succeq^* y$ if and only if $x \ge y$.
Let $\succsim_k$ be the preference defined by the piece-wise linear utility function $u_k$ with $u_k(0)=u_k(1)=1$ and $u_k(1/k)=0$; so, $u_k(x)=1-k x$ for $x \le 1/k$ and $u_k(x)=-1/(k-1) + k/(k-1) x$ for $x > 1/k$. Note that $\P$ includes $\succeq^*$ and all $\succsim_k$ for $k \ge 2$. 
Consider any set of growing experiments $\{ \Sigma_n \}$ in which the choice problems use non-zero alternatives. In every experiment $\Sigma_n$, Bob's choices, if Bob chooses according to his preference, can be rationalized with $\succsim_k$ for $k$ large enough. As $k$ grows large, $\succsim_k$ does converge, but to a preference distinct from $\succeq^*$, so that there is information regarding Bob's preference that is never fully learned. This problem owes to the fact that the set of locally strict preferences on $[0,1]$ is not closed: the limiting preference $\succsim_\infty$ is such that $x \succsim_\infty y$ if $x \ge y$ or if $x=0$, it is thus not locally strict. 

This simple example illustrates a more general fact. No matter the subject's underlying preference and the choice of growing experiments, one can always find locally strict preferences that perfectly rationalize the observations of the subject who chooses in accordance to his preference, and yet, as the size of the experiment grows large, converge to total indifference among all alternatives.\footnote{See Section~\ref{appsec:Estimation and closed preference sets} in the Online Appendix.} The rationalizing preferences then convey no information on the features of the subject's preference that are not directly observed. Therefore closedness is not a mere technical assumption. It is however satisfied in several important cases. In Sections~\ref{sec:Preferences from utilities} and~\ref{sec:Application to monotone preferences}, we establish that the relevant set locally strict preferences is closed in the major preference environments.

Finally, the choice problems in the  experiments must be sufficiently many, and sufficiently diverse, so that observed behavior on all these choice problems can effectively probe the subject's preference. A set of experiments $\{ \Sigma_n \}$, with $\Sigma_n = \{B_1, \dots, B_n\}$, is called \df{exhaustive} when it satisfies the following two properties:
\begin{enumerate}
	\item $\bigcup_{k=1}^\infty B_k$ is dense in $X$.
	\item For all $x, y \in \bigcup_{k=1}^\infty B_k$ with $x \ne y$, there exists $k$ such that $B_k = \{x,y\}$.
\end{enumerate}
The first property imposes that the alternatives that are used in the set of experiments sample the space of alternatives appropriately. The second property states that the experimenter should be able to elicit the subject's choices over all alternatives used in her experiments. Note that denseness is the only real constraint: starting from a countable dense set of alternatives, one can always construct an exhaustive set of experiments via routine diagonalization arguments.\footnote{In the more general case where choices are made over more than two proposed alternatives, the analog exhaustiveness property would require that any comparison between any two alternatives used in the set of experiments is eventually observed or inferred in a large enough experiment.}

\medskip
\begin{condrevealed}[resume]
	\item $\{\Sigma_n\}$ is exhaustive. \label{cond:3}
\end{condrevealed}
\medskip

The importance of having a dense set of alternatives is clear: without it, the characteristics of the preference remains unobservable on an open set, and for general classes of preferences, knowledge of the preference outside this set does not suffice to infer those unobservable characteristics. 

The importance of local strictness for our results hinges on the fact that, for an exhaustive set of experiments, and any two distinct locally strict preferences $\succeq_A$ and $\succeq_B$ of two subjects $A$ and $B$ respectively, there always is at least one experiment for which subject $A$ behaves differently from subject $B$, thereby allowing the experimenter to distinguish between these two preferences. Thus, with local strictness, a false hypothesis will eventually be demonstrated to be false, whereas without it, too many preferences can be consistent with the data. This fact is stated formally in Lemma~\ref{lem:1}. 

\begin{lemma}
	\label{lem:1}
	Consider an exhaustive set of experiments with binary choice problems $\{x_k, y_k\}$, $k \in \Na$. Let $\succeq$ be any complete binary relation, and $\succeq_A$ and $\succeq_B$ be locally strict preferences. If, for all $k$, $x_k \succeq_A y_k$ and $x_k \succeq_B y_k$ whenever $x_k \succeq y_k$, then ${\succeq}_A={\succeq}_B$.
\end{lemma}
The proof of Lemma~\ref{lem:1} is in Appendix~\ref{app:proof:lem-1}.

Under the above assumptions, we establish the convergence of rationalizing preference estimates.

\begin{theorem}
	\label{thm:1}
	Suppose the revealed preference model $(X,\P,\{\Sigma_n\})$ meets Assumptions~\ref{cond:1}--\ref{cond:3} and $c$ is an arbitrary choice function. 
	If, for every $n$, the preference ${\succeq}_n \in \P$ rationalizes the observed choices on $\Sigma_n$, then there exists a preference ${\succeq}^* \in \P$ such that ${\succeq}_n \rightarrow {\succeq}^*$. Moreover, the limiting preference is unique: if, for every $n$, ${\succeq}_n' \in P$ rationalizes the observed choices on $\Sigma_n$, then the same limit ${\succeq}_n' \rightarrow {\succeq}^*$ obtains.
\end{theorem}
The proof of Theorem~\ref{thm:1} is in Appendix~\ref{app:proof:thm-1}.

Theorem~\ref{thm:1} asserts that if, in each experiment, the data can be rationalized by some preference in the class $\P$, then there always exists one preference in $\P$ that rationalizes the choices made over \emph{all} the experiments, and most importantly, there exists only one such preference.
The observations are exactly as if the subject's choices were guided by this particular preference, which can be obtained as the limit of the rationalizations as experiments grow in size. 

In particular, if we postulate the existence of a preference $\succeq^* \in \P$ according to which the subject chooses on any given decision problem, then no matter the selection of the rationalizing preferences, they always converge to $\succeq^*$.

\begin{corollary}
	\label{cor:1} 
	Suppose the revealed preference model $(X,\P,\{\Sigma_n\})$ meets Assumptions~\ref{cond:1}--\ref{cond:3} and $c$ is an arbitrary choice function.
	If the preference ${\succeq}^* \in \P$ rationalizes $c$ and if, for every $n$, the preference ${\succeq}_n \in \P$ rationalizes the observed choices on $\Sigma_n$, then ${\succeq}_n \rightarrow {\succeq}^*$.
\end{corollary}

\subsection{Convergence in statistical preference models.}\label{subsec:Convergence in statistical preference models}

When the subject makes mistakes, looking for a preference that perfectly rationalizes his behavior is moot---a rationalizing preference in the class $\P$ may not exist. Instead, we introduce a simple estimator that approximately rationalizes the observed data, based on minimization of the Kemeny distance \citep{kendall1938,kemeny1959}. 

The estimator results from a two-step procedure. Let us look at an experiment of size $n$, $\Sigma_n$, drawn at random according to the experimental design. Let $c$ be the choice function that captures the choices of the subject, which is also random. 
First, from the choices observed on $\Sigma_n$, a revealed preference relation is constructed that captures these choices exactly. 
This revealed preference relation, denoted $R_n$, is defined by $x \mathrel{R}_n y$ for all $\{x,y\} \in \Sigma_n$ such that $c(\{x,y\}) = x$---that is, $x \mathrel{R}_n y$ when the subject chooses $x$ in the choice problem $\{x,y\}$. Note that $R_n$ is sparse, as it only conveys information on the alternatives used in $\Sigma_n$. 
Secondly, the estimated preference ${\succeq}_n$ is chosen to minimize the distance $d_n({\succeq}, R_n)$ between the revealed preference relation just defined, and a preference in ${\succeq}\in \P$; 
\begin{equation*}
	{\succeq}_n \in \argmin \{ d_n({\succeq}, R_n) : {\succeq} \in \P \}.
\end{equation*} Distance $d_n$ is taken to be a version of the Kemeny distance, a rank distance measure, defined by 
\begin{equation*}
	d_n ({\succeq},R_n) = \frac{1}{n} \abs{R_n \setminus {\succeq}}.	
\end{equation*}
In words, $d_n({\succeq}, R_n)$ averages the number of mistakes made by the subject on $\Sigma_n$ if his underlying preference is $\succeq$. We refer to this estimator as the \df{Kemeny-minimizing estimator}.\footnote{\label{fn:Kemenydefn}$\abs{R_n \setminus {\succeq}}$ denotes the number of elements in $R_n \setminus {\succeq}$. The Kemeny distance between two finite binary relations $R$ and $R'$ is usually defined as $\abs{R \mathrel{\Delta} R'}$, where $\Delta$ is the symmetric difference. 
Note that if $(x,y)\in {\succeq}\setminus R_n$ and $\{x,y\}\in \Sigma_k$, Then
$(y,x)\in R_n\setminus {\succ}$.  In our model, alternatives are
strictly ranked by $R_n$ and by $\succeq$ with probability one.
Hence, 
\begin{multline*}
\sum_{\{x,y\}\in \Sigma_n} \big(
\one_{(x,y)\in {\succeq}\setminus R_n} + \one_{(x,y)\in R_n\setminus {\succeq}}+
\one_{(y,x)\in {\succeq}\setminus R_n} + \one_{(y,x)\in R_n\setminus {\succeq}}
\big)
\\
= 2 \sum_{\{x,y\}\in \Sigma_n} \big(\one_{(x,y)\in R_n\setminus {\succeq}} +
\one_{(y,x)\in R_n\setminus {\succeq}}\big)    
\end{multline*}
with probability one, which justifies our Kemeny distance terminology.}

\newlist{condstat}{enumerate}{1}
\setlist[condstat]{label=\scshape Assumption \arabic*'.~~, ref=(\arabic*'), wide=25pt,widest=99,leftmargin=*}

For a statistical preference model $(X,\P,\la,q)$, we shall need three assumptions. Assumptions~\ref{cond:1} and~\ref{cond:2} on the preference environment $(X,\P)$ remain unchanged. In particular, we continue to assume that the preferences under consideration are locally strict, local strictness being the key unifying property between the revealed and statistical preference models.

We think of Assumption~\ref{cond:3S}, below, as the analog of Assumption~\ref{cond:3} for randomized experiments. Its main import is that  we almost never draw a decision problem that makes the subject indifferent. It also imposes that the sampling distribution have full support.\footnote{Full support means that there is no proper closed subset of the sample space that has probability 1.} Analogously to Assumption~\ref{cond:3}, this ensures that enough binary choice problems are probed. Full support can be relaxed for preferences that are identified on a proper subset of $X$ (as we do in Section~\ref{subsec:Commodity spaces}).

\medskip
\begin{condstat}[resume]
	\setcounter{condstati}{2}
	\item $\lambda$ has full support on $X$ and, for all ${\succeq} \in \P$, $\{(x,y) : x \sim y\}$ has $\lambda$-probability 0.\label{cond:3S}
\end{condstat}
\medskip

Therefore, under Assumption~\ref{cond:3S}, using the same binary choice problem twice in an experiment occurs with probability zero. Together with Assumption~\ref{cond:2}, Assumption~\ref{cond:3S} guarantees identification in the usual sense: the distribution of the data under the true preference is different than that at any other preference.\footnote{Identification in the usual sense is implied by the identification condition for consistency, proved in Lemma~\ref{lem:uniquemax} of Appendix~\ref{app:proof:thm-2}. Although Assumption~\ref{cond:3S} is crucial, it is by itself is not sufficient; for example, two preferences that differ at one point only cannot be distinguished when $\lambda$ has full support. This case is ruled out by local strictness.} 
This assumption allows to control the problem of partial observability. If instead indifference were to occur with positive probability, then the model would only be partially identified. 

Under the above assumptions, the Kemeny-minimizing estimator is consistent. Recall that $\rho$ denotes the metric on the space of preferences. 

\begin{theorem}
	\label{thm:2}
	Suppose the statistical preference model $(X,\P,\la,q)$ meets Assumptions~\ref{cond:1},~\ref{cond:2} and \ref{cond:3S}, and suppose the subject's preference is ${\succeq}^* \in \P$.
	Let ${\succeq}_n$ denote the Kemeny-minimizing estimator for the $n$-th experiment $\Sigma_n$. Then, $\{ {\succeq}_n \}_{n \in \Na}$ converges to ${\succeq}^*$ in probability; that is, for any $\eta>0$, 
	\[ \lim_{n \rightarrow \infty} \mathbf{Pr}\big( \rho({\succeq}_n,{\succeq}^*) < \eta \big) = 1.\]
\end{theorem}
The proof of Theorem~\ref{thm:2} is in Appendix~\ref{app:proof:thm-2}. One challenge is that the class of preferences may be very rich, which increases the potential for overfitting: the noise in the data may be misinterpreted as being part of the subject's true preference. And indeed, for many common spaces of alternatives, it can be shown that one can rationalize perfectly any finite set of observations by a locally strict preference. Imposing that the class of preferences be closed allows to overcome this difficulty.

We stress that Theorem~\ref{thm:2} requires \emph{no assumption} on the error probability function $q$, other than measurability and asking that the subject be more likely to follow his preference than to make a mistake. Alice may remain agnostic about the dependence of $q$ on the underlying preference $\succeq$ and the alternatives to choose from $x$ and $y$. Moreover, aside from the independence of the draws of alternatives, the requisites on $\la$, stated in Assumption~\ref{cond:3S}, are minimal. In particular, calculating the Kemeny estimator does not require any assumptions on $q$ and $\la$. It requires Alice to specify $\P$, but allows her to be largely agnostic about the rest of the model.


Having the guarantee that preference estimates converge accurately, Alice may want to know how large of a sample is needed to estimate the subject's preference within a given approximation error. Our third result establishes lower bounds on the rate of convergence. 

To state the result, we introduce some terminology. For any $\eta > 0$ and any $\delta \in (0,1)$, let $N(\eta,\da)$ be the smallest value of $N$ such that for all $n \ge N$, and all underlying subject preferences ${\succeq}^* \in \P$,
\[
	\mathbf{Pr}( \rho({\succeq}^n,{\succeq}^*) < \eta ) \geq  1-\da.
\]
(By convention, we let $N(\eta,\da) = \infty$ if no finite value of $N$ exists.) That is, $N(\eta,\da)$ is the size of the smallest experiment such that the probability that the preference estimate is $\eta$-close to the true subject preference is guaranteed to be at least $1-\da$, no matter the true subject preference.

In addition, $\mu(.,{\succeq}^*)$ denotes the probability measure induced on the product space $X\times X$ by $\succeq^*$, $q$ and $\la$, as follows:
\[ 
	\mu(A,{\succeq}^*) = \int_A q({\succeq}^*;x,y) \de \la(x,y).
\]
Loosely speaking, $\mu(x,y,{\succeq}^*)$ represents how likely a subject with preference $\succeq^*$ is to choose $x$ over $y$ in a decision problem randomly drawn. In particular, the value of $\mu({\succeq},{\succeq}^*)$ represents the probability that the choice of a subject with preference $\succeq^*$ over a randomly drawn decision problem is consistent with the preference $\succeq$.
 
Finally, for $\eta > 0$, we let
\begin{equation}
	\label{eq:eta-definition}
	r(\eta) = \inf \big\{ \mu({\succeq}^*,{\succeq}^*) - \mu({\succeq},{\succeq}^*) : {\succeq},{\succeq}^* \in \P, \rho({\succeq}^*, {\succeq}) \ge \eta \big\}.	
\end{equation}
Roughly, the value of $r(\eta)$ captures the smallest possible probability that a subject make a choice that is perfectly consistent with his own preference but is inconsistent with a preference at least $\eta$-distant.

To obtain convergence rates, we appeal to Vapnik-Chervonenkis theory.\label{R2:VCpage} Given a preference environment $(X,\P)$, for $n \in \Na$, let $S_n$ be the largest size of the sets
\[
    \big\{ ( \one_{x_1 \succeq y_1}, \dots,  \one_{x_n \succeq y_n}) \in \{0,1\}^n : {\succeq} \in \P \big\}
\]
over all binary choice problems $\{ x_k, y_k \} \subset X$ for $k=1,\dots,n$. 
We always have $S_n \le 2^n$, and, if $\P$ is rich enough, we may have $S_n = 2^n$. 
The VC dimension of $\P$ (abbreviation for Vapnik-–Chervonenkis dimension) is then defined as the maximum value of $n$ such that $S_n = 2^n$, and is infinite if no such maximum exists. 

The consistency of the Kemeny-minimizing estimator applies no matter the VC dimension of $\P$, but when $\P$ has a finite VC dimension, and so is not too rich, we can, in addition, obtain uniform bounds on the convergence rates.

\begin{theorem}
	\label{thm:3}
	If the statistical preference model $(X,\P,\la,q)$ meets Assumptions~\ref{cond:1},~\ref{cond:2},and~\ref{cond:3S}, and if $\P$ has a finite VC dimension, then\footnote{The big O notation refers to the usual asymptotic upper bound.} 
	\[
        N(\eta, \delta) = O\left(\frac{1}{r(\eta)^2} \ln \frac{1}{\delta} \right).
    \]
\end{theorem}

The proof of Theorem~\ref{thm:3} is in Appendix~\ref{app:proof:thm-3}, in which we also provide a more refined nonasymptotic bound. 
Of course, the value of $r(\eta)$ depends on the specific preference environment being considered. Below in Sections~\ref{sec:Preferences from utilities} and~\ref{sec:Application to monotone preferences}, we apply Theorem~\ref{thm:3} in different environments. \citet*{basu2018learnability} compute the VC dimension of some common models of choice, they show, in particular, that the class of expected utility, Choquet expected utility, and two-state max-min preferences have finite VC dimension.

\section{Preferences from utilities}\label{sec:Preferences from utilities}

In this section and the next, we show that the assumptions of our general framework are valid in a variety of important preference environments.
This section handles preference environments derived from collections of utility functions. 
Specifically, we show how our assumptions may be derived from conditions on the utility representations under consideration, rather than directly imposing the assumptions on a family of preferences.

Consider a set of alternatives $X$. We are interested in sets of preferences $\P$ that correspond to sets of utility functions.
A \df{utility function} is any function $u : X \rightarrow \Re$. We endow the space of utility functions with the topology of compact convergence.\footnote{A sequence of functions $f_n : X \rightarrow \Re$, $n \in \Na$, converges compactly to a function $f$ if and only if it converges uniformly to $f$ on every compact set $K \subseteq X$. This topology on utilities is commonly used in the literature; see for example \citet*{mas1974continuous} and \citet*{border1994dynamic}.} For any utility function $u : X \rightarrow \Re$, let $\Phi(u)$ denote the preference induced by $u$, that is, the binary relation defined by $x \mathrel{\Phi(u)} y$ if and only if $u(x) \geq u(y)$. And for any set of utility functions $\U$, let $\Phi(\U) = \{ \Phi(u) : u \in \U \}$ denote the image of $\U$. 

The following proposition states conditions on the set of utility functions being considered under which our main results apply.

\begin{proposition}
	\label{prop:utility-revealed}
	Suppose $X$ satisfies Assumption~\ref{cond:1}, $\U$ is a compact set of continuous utility functions, and every preference in $\Phi(\U)$ is locally strict. Then the class of preferences $\P \equiv \Phi(\U)$ meets Assumption~\ref{cond:2}.
\end{proposition}

Proposition~\ref{prop:utility-revealed} is an immediate implication of Theorem 8 of \citet{border1994dynamic}, who establish the continuity of $\Phi$ (see Appendix~\ref{app:proof:utility-revealed}).

\subsection{Intertemporal consumption.}\label{subsec:Intertemporal consumption}

To put Proposition~\ref{prop:utility-revealed} to work in a simple concrete example, let us look at a case of intertemporal choice. 
Suppose a good can be consumed at $d \ge 2$ different dates $t_1 < \dots < t_d$. In this environment, an alternative is a vector of the Euclidean space $\Re_{+}^d$ whose $i$-th entry indicates the amount consumed at the date $t_i$.

Fix $a, b \in \Re_{++}$ with $a < b$ and call $\mathcal{V}$ the set of continuous functions $v:\Re_+ \rightarrow \Re$ that satisfy, for all $x < y$,
\begin{equation*}
	a \cdot (y-x) \le v(y) - v(x) \le b \cdot (y-x).
\end{equation*}
We interpret $v(x)$ as the utility for an immediate consumption of quantity $x$ of the good. The above inequality constrains marginal utilities to be positive and bounded above and below.

Denote by $\U$ the set of the utility functions $u$ over $\Re_{+}^d$ that are written
\begin{equation*}
	u(x_1, \dots, x_d) = \sum_{i=1}^d \delta_i v(x_i),
\end{equation*}
where $v \in \mathcal{V}$, and $\delta = (\delta_1, \dots, \delta_d) \in [\varepsilon,1]^d$ is a vector of discount factors, with $\varepsilon$ an arbitrarily small positive lower bound. So, the set $\U$ captures discounted utility preferences with general discount factors.

Compactness of $\U$ follows directly from the Arzel\`{a}-Ascoli theorem (for example, Theorem 6.4 of \citealp{dugundji}). And clearly each utility function describes a locally strict preference, because if an individual with utility $u \in \U$ prefers the consumption vector $x \in \Re_+^d$ to $y \in \Re_+^d$, then he strictly prefers the consumption vector $x + \eta \one$ to $y$, for any $\eta > 0$. 
Therefore, we have the following corollary to Proposition~\ref{prop:utility-revealed}.

\begin{corollary}
	\label{cor:intertemporal-consumption} 
	In the intertemporal-consumption environment just described, the set of alternatives $X \equiv \Re_{+}^d$ endowed with the Euclidean topology meets Assumption~\ref{cond:1}, and the class of preferences $\P$ that is induced by the utility functions in $\U$ meets Assumption~\ref{cond:2}.
\end{corollary}

For example, one possible use of our theory is to recover discount factors from the data, or to check for distortions with respect to standard models such as exponential discounting.

\subsection{Expected utility preferences.}\label{subsec:Expected utility preferences}

Expected utility preferences are an important case of preferences derived from utility functions.
Let $\Pi \equiv \{\pi_1, \ldots, \pi_d \}$ be a collection of $d \ge 2$ prizes, and let $\Delta^{d-1}$ denote the $(d-1)$-dimensional simplex $\{ p \in \Re^d_+ : p_1 + \cdots + p_d = 1\}$. Think of each element $p$ of the simplex as a lottery over the prizes in $\Pi$, with $p_i$ the probability of getting $\pi_i$. The set of alternatives is $\Delta^{d-1}$, endowed with the Euclidean topology.

An \df{expected utility preference} stands for any preference $\succeq$ on $\Delta^{d-1}$ defined by a vector of utility indices $v \in \Re^d$, with the property that $p \succeq p'$ if and only if $v \cdot p \geq v \cdot p'$, where $v \cdot p = \sum_{i=1}^d v_i p_i$ is the expected utility of lottery $p$. This preference is \df{nonconstant} if there is at least one pair $p, p' \in \Delta^{d-1}$ for which $p \succ p'$, or equivalently, if the vector of utility indices satisfies $v_i \ne v_j$ for some $i, j$.
Of course, the vector of utility indices that defines a nonconstant expected utility preference is not unique. 
To ensure uniqueness, we  normalize the utility indices by imposing that indices sum to zero and $v$ be on the unit sphere. That is, each preference of $\P$ is uniquely associated to a normalized vector of utility indices in the set\footnote{$\norm{\cdot}$ denotes the Euclidean norm $\norm{x - y} \equiv \sqrt{\sum_{i=1}^d (x_i - y_i)^2}$.}
\begin{equation*}
	\pc{v \in \Re^d : \sum_{i=1}^d v_i=0, \, \|v\|=1}.
\end{equation*}

The following result asserts that the standard expected utility model is situated within our framework.


\begin{corollary}
	\label{cor:expected-utility-revealed} 
	In the expected-utility environment just described, the set of alternatives $X \equiv \Delta^{d-1}$ endowed with the Euclidean topology meets Assumption~\ref{cond:1}, and the class $\P$ of all nonconstant expected utility preferences meets Assumption~\ref{cond:2}.
\end{corollary} 
The proof of Corollary~\ref{cor:expected-utility-revealed} is in Appendix~\ref{app:proof:expected-utility-revealed}.

Note that, although we gather data over the entire set of relevant alternatives---here the whole simplex---we could use a smaller set, because expected utility preferences are identified on a small set of lotteries. By learning the preference on this smaller set, we can infer uniquely the preference on the full set.
For example, suppose the alternatives that make the set of binary choice problems are from a subset of the simplex $\Delta^{d-1}$ that is convex and compact with non-empty interior, and that the family of binary choice problems are exhaustive relative to that subset. Then, Theorem~\ref{thm:1} continues to hold even though preferences continue to be defined over $\Delta^{d-1}$, that is, the rationalizing preference estimates continue to converge to a unique limiting preference. A similar observation applies to Theorems~\ref{thm:2} and~\ref{thm:3}.

In the expected-utility model, we can further refine the convergence rates of Theorem~\ref{thm:3}, provided that we restrict attention to error probability functions $q$ that are polynomially bounded in the following sense: there exists $C > 0$ and $k > 0$ such that,
if lottery $p$ is strictly preferred to lottery $p'$ according to preference $\succeq$, then
\begin{equation}
	\label{eq:q-restricted-1}
	q({\succeq};p,p') \ge \frac{1}{2} + C |u \cdot p - u \cdot p'|^k,
\end{equation}
where $u$ is the normalized vector of utility indices associated with preference $\succeq$. Equation~\eqref{eq:q-restricted-1} is satisfied, for example, if $q({\succeq};p,p')$ is lower-bounded strictly above $1/2$ over all $\succeq$ and all $p,p'$ such that $p \succ p'$. It is also satisfied if the probability of making an error decreases with the difference of utilities between the two lotteries, for example if, for a continuously differentiable and nondecreasing function $f:\Re_{+} \rightarrow \Re_{+}$ with $f'(0) > 0$, we can write $q({\succeq};p,p') = 1/2 + f(u \cdot p - u \cdot p')$ for all $\succeq$ and all $p,p'$ such that $p \succ p'$ ($u$ continues to denote the utility indices associated $\succeq$). In the former case, we can use $k=0$, and in the latter case, $k=1$. 

Applying Theorem~\ref{thm:3} also requires to specify a compatible metric on preferences.
Because the set of lotteries is compact and metrizable, one could measure the distance between two preferences by the Hausdorff metric.
Here however, each preference is naturally represented by a finite-dimensional vector of utility indices, and it can be more straightforward to set, as distance between preferences, the distance between their associated utility indices. 
So, let $\rho({\succeq}, {\succeq}')$ be the Euclidean distance between the two normalized vectors of utility indices of ${\succeq}$ and ${\succeq}'$. It can be seen that $\rho$ is a compatible metric (this fact holds quite generally, see Section~\ref{subsec:On parameterized sets of preferences}). 

The focus on error probability functions of the above form then allows for explicit convergence rates for the Kemeny-minimizing estimator, as below.
\begin{corollary}
	\label{cor:expected-utility-stat} 
	For the statistical preference model $(X,\P,\la,q)$, where $X \equiv \Delta^{d-1}$, $\P$ is the set of all nonconstant expected utility preferences, $\la$ is the uniform distribution on $\Delta^{d-1}$, and $q$ satisfies Equation~\eqref{eq:q-restricted-1}, the Kemeny-minimizing estimator is consistent and, as  $\eta \to 0$ and $\da \to 0$,
	\begin{equation*}
		N(\eta,\da) = O\pp{ \frac{1}{\eta^{8(d-1) + 4k} } \ln \frac{1}{\da} }.
	\end{equation*}
\end{corollary} 
One can rewrite Corollary~\ref{cor:expected-utility-stat} to provide a convergence rate of the form $O_p((1/n)^{1/d})$.\footnote{The $O_p$ notation refers to the stochastic boundedness notation.} The proof of Corollary~\ref{cor:expected-utility-stat} is in Appendix~\ref{app:proof:expected-utility-stat}.

Note that the uniform distribution is not at all essential for Corollary~\ref{cor:expected-utility-stat}. It just yields a particularly simple closed form for the bound on $r(\eta)$ used in Theorem~\ref{thm:3}.

\section{Application to monotone preferences}\label{sec:Application to monotone preferences}

In many economic settings, it is safe to posit the existence of a universal ordering, by which some alternatives are ranked above others by all the individuals of the relevant population: preferences are monotone with respect to this ordering.  For example, for the classical consumption environment in which individuals choose bundles of goods, it is usually assumed that individuals strictly prefer to have more of each good. In laboratory experiments, it is common to assume some form of objective ranking, for instance when enforcing single-switching in price lists, or when using randomization devices to enforce incentives.
Monotonicity with respect to such universal orderings turns out to be a very useful discipline on preferences. 

In this section, we adopt the following terminology. Fix a set of alternatives $X$. We call \df{dominance relation} any binary relation $\rhd$ on $X$ that is not reflexive, that is, for each $x \in X$, $x \ntriangleright x$. 
The relation $\rhd$ is said to be \df{open} when $\rhd$ is an open set in the product space $X \times X$.\footnote{That is, for each $x, y$ with $x \rhd y$, there exists a neighborhood $V$ of $(x,y)$ in $X \times X$ such that for all $(x',y') \in V$, $x' \rhd y'$.}
Being open for a dominance relation can be interpreted as a continuity property, saying that if $x$ dominates $y$ then this domination extends locally around the alternatives $x$ and $y$.

Given a dominance relation $\rhd$, a preference relation $\succeq$ is \df{strictly monotone} with respect to $\rhd$ if, for each $x,y \in X$, $x \rhd y$ implies $x \succ y$. Having a class of strictly monotone preferences captures the above idea that some alternatives are universally preferred to some others in accordance to the dominance relation.

Usually, strict monotonicity alone does not suffice to ensure that the preference is locally strict, the first crucial condition in our framework. It helps to add a notion of transitivity. We call a preference relation $\succeq$ \df{Grodal-transitive} if for all $x,y,z,w \in X$, $x \succeq y \succ z \succeq w$ implies $x \succeq w$. Named after \citet{grodal1974note}, Grodal-transitivity is weaker, and so more permissive, than the usual notion of transitivity. Importantly, together with strict monotonicity, Grodal-transitivity makes the class of preferences closed, the second crucial condition of our framework.

\begin{lemma}
	\label{lem:monotone-closed} 
	Suppose Assumption~\ref{cond:1} is met and $\rhd$ is a dominance relation on $X$. If $\rhd$ is open, then the class of preferences that are Grodal-transitive  and strictly monotone with respect to $\rhd$ is closed.
\end{lemma}
The proof of Lemma~\ref{lem:monotone-closed} is in Appendix~\ref{app:proof:monotone-closed}.

It is worth noting that, in general, closedness is not achieved under the usual notion of transitivity and strict monotonicity. If one wishes to impose transitivity, the class of preferences must be reduced further to obtain a closed set (as we did in the example of Section~\ref{subsec:Expected utility preferences}). Of course, there is no harm in having a more generous class of preferences. Even if the class includes preferences that fails desirable properties such as classical transitivity---and so may include irrelevant preferences---preference estimates are guaranteed to converge to the correct underlying preference, so that any violation by the preference estimates eventually gets corrected in the limit. 

The main benefit of Grodal-transitivity is that it is enough to make strictly monotone preferences locally strict under relatively mild conditions.

\begin{lemma}
	\label{lem:monotone-locallystrict} 
	Suppose Assumption~\ref{cond:1} is satisfied and $\succeq$ is a preference strictly monotone with respect to the dominance relation $\rhd$. If, for each $x \in X$, there exists $y, z \in X$ arbitrarily close to $x$ and such that $y \rhd x$ and $x \rhd z$, then $\succeq$ is locally strict. 
\end{lemma}
The proof of Lemma~\ref{lem:monotone-locallystrict} is in Appendix~\ref{app:proof:monotone-locallystrict}.

Therefore, Assumption~\ref{cond:2} on the class of preferences is valid as long as $X$ is well behaved, the dominance relation is open, and the preferences considered are Grodal-transitive and strictly monotone. The examples below show that these properties are satisfied in many common preference environments.

Our discussion is summed up in the following result, which we shall see has several applications. 

\begin{proposition}
	\label{prop:monotone-preferences}
	Suppose that $X$ satisfies Assumption~\ref{cond:1}, and is endowed with an open dominance relation $\rhd$. Suppose that $\P$ consists of the preferences that are Grodal-transitive and strictly monotone with respect to $\rhd$. Suppose further that for each $\succeq\in \P$ and $x \in X$, there exists $y, z \in X$ arbitrarily close to $x$ such that $y \rhd x$ and $x \rhd z$. Then $\P$ meets Assumption~\ref{cond:2}. 
\end{proposition} 

\subsection{Commodity spaces.}\label{subsec:Commodity spaces}

The classical setup of consumer demand analysis features a commodity space over $d \ge 2$ goods, where consumers get to choose over bundles of goods (as in \citealp{afriat}, \citealp{mascolell78}, or \citealp{varian1982nonparametric}). The set of alternatives is the Euclidean space $\Re_{++}^d$, the $i$-th entry of vector $(x_1, \dots, x_d)$ is interpreted as the consumed quantity of the $i$-th good. This environment is part of our framework when preferences are asked to satisfy a monotonicity condition.

Consider the dominance relation $\gg$ on $\Re_{++}^d$ by $x \gg y$ exactly when $x_i > y_i$ for all $i = 1, \dots, d$. An individual whose preference is strictly monotone with respect to $\gg$ means that this individual strictly prefers to have more of every good, a postulate that appears reasonable in many situations, and that is common in economic models.
It is evident that the relation $\gg$ is open, and for any $x \in \Re_{++}^d$, $x + \varepsilon \one \gg x$ for all $\varepsilon > 0$ while $x \gg x - \varepsilon \one \in \Re_{++}^d$ for all small enough $\varepsilon > 0$. Hence, Lemmas~\ref{lem:monotone-closed} and~\ref{lem:monotone-locallystrict} apply, and we get Proposition~\ref{prop:commodity-spaces-revealed}.\footnote{While the set $\Re_{++}^d$ is not complete under the Euclidean metric, there exists a compatible complete metric by Alexandroff's Theorem (Theorem 24.12 of \citealp{willard}). Of course, $\Re_{++}^d$ is also locally compact and separable, and hence Assumption~\ref{cond:1} is satisfied.} 

\begin{corollary}
	\label{prop:commodity-spaces-revealed} 
	In the commodity-space environment just described, the set of alternatives $X \equiv \Re_{++}^d$ endowed with the Euclidean topology meets Assumption~\ref{cond:1}, and the class $\P$ of all preferences that are Grodal-transitive and strictly monotone with respect to $\gg$ meets Assumption~\ref{cond:2}. 
\end{corollary} 

The same set of alternatives can be used to describe state-contingent payments, with an objective public distribution over states, and where the $i$-th entry of a vector encodes the payment received in the state $i$. In such an environment, one may want to test the validity of the hypothesis that individuals maximize an expected utility function (as, for example, in \citealp{green1986expected}), or maximize a utility function that is monotone with respect to first-order stochastic dominance (as in \citealp*{nishimura}). Because both classes of preferences are more restrictive than the class $\P$ considered here, our convergence results continue to apply, which means that we can fully recover preferences and examine precisely the validity of these hypotheses.

\subsection{Choice over menus.}\label{subsec:Choice over menus}

Our next application deals with recovering preferences over menus, following \citet{kreps1979representation} and \citet*{dlr}. Let $\Pi=\{\pi_1, \dots, \pi_d\}$ be a collection of prizes and let $\Delta_{++}^{d-1}$ be the interior of the $(d-1)$-dimensional simplex, interpreted as the set of full-support distributions over the elements of $\Pi$. We endow $\Delta_{++}^{d-1}$ with the Euclidean metric. 

Let $\M$ denote the set of closed convex subsets of $\Delta_{++}^{d-1}$ with nonempty interior. We interpret $\M$ as a set of menus of lotteries.  A subject who possesses a menu $m \in \M$ gets to choose a lottery in $m$, and subsequently receives a prize drawn according to this lottery. The convexity of menus that is assumed here is also implied by the axiom of indifference to randomization introduced by \citet{dlr}. We endow $\M$ with the Hausdorff topology, as is standard in menu theory.  

We define the dominance relation $\sqsupset$ as follows: for two menus $m_A$ and $m_B$, $m_A \sqsupset m_B$ if every expected-utility decision maker with full knowledge of her utility when making menu choices would strictly prefer $m_A$ to $m_B$. More precisely, we write
\begin{equation*}
	\U \equiv \left\{ u \in \Re^d : \sum_{i=1}^d u_i = 0, \|u\|=1 \right\}
\end{equation*}
the set of all utility indexes over prizes, up to a normalization (we rule out the trivial preference that is indifferent between any two lotteries). Then we write $m_A \sqsupset m_B$ if and only if for every $u \in \U$, 
\begin{equation*}
	\sup_{p\in m_A} u \cdot p > \sup_{p \in m_B} u \cdot p.
\end{equation*}
Since we restrict attention to convex menus, $A \sqsupset B$ implies $A \supset B$. Hence, the dominance relation $\sqsupset$ is similar to, but weaker than, the set-containment relation traditionally used in models of choice over menus. In particular, monotonicity with respect to $\sqsupset$ is less demanding than monotonicity with respect to $\supset$.

Our next result establishes that the revealed preference framework applies to the menu preference environment.

\begin{corollary}
	\label{cor:preferences-menus-revealed} 
	In the menu environment just described, the set of alternatives $X \equiv \M$ endowed with the Hausdorff topology meets Assumption~\ref{cond:1}, and the class $\P$  of all preferences that are Grodal-transitive and strictly monotone with respect to $\sqsupset$ meets Assumption~\ref{cond:2}.
\end{corollary}
	
The proof of Corollary~\ref{cor:preferences-menus-revealed} is in Appendix~\ref{app:proof:preferences-menus-revealed}.

\subsection{Intertemporal consumption revisited.}\label{subsec:Intertemporal consumption revisited}

\newcommand{\symm}{\mathrel{>}_{\textrm{sym}}}

We revisit the example in Section~\ref{subsec:Intertemporal consumption}. There is a good to be consumed at a sequence of $d \ge 2$ increasing dates $t_1, \dots, t_d$. The set of alternatives is now $\Re_{++}^d$, and each element $(x_1, \dots, x_d)$ gives the amount consumed at each date.

As before, we can hypothesize that individuals prefer more of the good early over less later. In the present environment, this postulate is captured by the dominance relation $\ggg$ on $\Re_{++}^d$ whereby $x \ggg y$ if and only if for every $k$,
\begin{equation*}
	\sum_{i=1}^k x_i > \sum_{i=1}^k y_i.
\end{equation*} This dominance relation also captures the impatience assumption implicit in exponential discounting.
In the same environment, \citet*{nishimura} suggest another postulate: that individuals are neutral to time while they still prefer to get more of the good. The associated dominance relation $\symm$ is defined as $x \symm y$ if and only if there exists a permutation $\sigma$ over $\{1, \dots, d\}$ such that for every $i$, $x_{\sigma(i)} > y_{\sigma(i)}$.

It is immediately seen that $\ggg$ is open, and that, when $x \in \Re_{++}^d$ and $\varepsilon > 0$ is small enough, $x \ggg x - \varepsilon \one$ and $x + \varepsilon \one \ggg x$. The very same observations apply to the relation $\symm$. By a logic that is now routine, we get the following proposition.

\begin{corollary}
	\label{cor:intertemporal} 
	In the intertemporal consumption environment just described, the set of alternatives $X \equiv \Re_{++}^d$ endowed with the Euclidean topology meets Assumption~\ref{cond:1}, and the class $\P$ of all preferences that are Grodal-transitive and strictly monotone with respect to $\ggg$ (or with respect to $\symm$) meets Assumption~\ref{cond:2}.
\end{corollary}

\subsection{Choice over lotteries.}\label{subsec:Choice over lotteries}

\newcommand{\fsd}{\mathrel{>}_{\textrm{FSD}}}

Consider a set of $d \ge 2$ monetary rewards $\Pi \equiv \{\pi_1, \dots, \pi_d\}$, and let the interior of the $(d-1)$-simplex, $\Delta_{++}^{d-1}$, be the set of alternatives endowed with the Euclidean topology. We interpret an element $p \in \Delta_{++}^{d-1}$ as a full-support lottery over monetary rewards. This choice domain is an instance of the domain studied in Section~\ref{subsec:Expected utility preferences}.

Suppose that the elements of $\Pi$ are ordered as $\pi_1 < \dots < \pi_d$. A natural dominance relation is strict first-order stochastic dominance, noted $\fsd$, where $p \fsd p'$ if and only if for all $k = 1, \dots, d-1$,
\begin{equation*}
	\sum_{i=1}^k p_i < \sum_{i=1}^k p_i'.
\end{equation*}
The reason for using \emph{strict} first-order stochastic dominance is that this relation is open, as opposed to first-order stochastic dominance. Now, let $p \in \Delta_{++}^{d-1}$. For a small enough positive $\varepsilon$, we can define $p' \in \Delta_{++}^{d-1}$ by $p_i' = p_i - \varepsilon$ for all $i \le d-1$ and $p_d' = p_d + (d-1) \varepsilon$. Then $p' \fsd p$. And similarly, when instead $p_i' = p_i + \varepsilon$ for all $i \le d-1$ and $p_d' = p_d - (d-1) \varepsilon$, $p \fsd p'$.  As a consequence of Proposition~\ref{prop:monotone-preferences}, we obtain the following result.

\begin{corollary}
	\label{cor:choice-over-lotteries} 
	In the lottery  environment just described, the set of alternatives $X \equiv \Delta_{++}^{d-1}$ endowed with the Euclidean topology meets Assumption~\ref{cond:1}, and the class $\P$ of all preferences that are Grodal-transitive and strictly monotone with respect to $\fsd$ meets Assumption~\ref{cond:2}.
\end{corollary}

Note that the class of nonconstant expected-utility preferences studied in Section~\ref{subsec:Expected utility preferences} and the class considered in the present lottery choice environment are distinct, and neither one is a refinement of the other.

\section{Concluding discussion}\label{sec:Concluding discussion}

This paper deals with the question of recovering individual preferences from observed choice data, when the data consist of a finite number of binary comparisons. Decision theorists often consider the question of ``backing out'' a model from data on pairwise choices, motivated by laboratory experiments, in which pairwise choices are common. They assume, however, rich and infinite data sets.  Econometricians study the convergence of preference estimates, but their models usually differ from the pairwise choice paradigm. Moreover, the conditions needed for consistency of their estimates are imposed as added-on assumptions, instead of being derived from the properties of the economic model under consideration. 

We provide a common unifying framework. We show that the class of preferences considered should be locally strict. Under local strictness, and some regularity conditions, any preference that rationalizes the observed pairwise choices converges to the correct data-generating preference. In the statistical counterpart to our model, the  Kemeny-minimizing estimator, which outputs the preferences that best fit the data, is consistent.  In addition, convergence rates can be obtained while remaining largely agnostic on the sampling method and the error probability function. 

Our results require weak assumptions and apply to a broad range of standard preference environments. We conclude with a discussion on a few aspects of our model.

\subsection{On the convergence of preferences.}\label{subsec:On the convergence of preferences}

At a general level, the topology of closed convergence is defined by the property that individuals with comparable preferences behave similarly on closely related decision problems. Such continuity property appears natural, and even necessary to be able to learn from finite data. 
In formal terms, if $x \succ y$ for some alternatives $(x,y)$ and a preference ${\succeq}$, and if $(x', y')$ are alternatives in a neighborhood of $(x,y)$, then the topology of closed convergence is defined exactly so that $x' \succ' y'$ for any preference ${\succeq'}$ close enough to ${\succeq}$.\footnote{To be even more formal, under the assumptions of our results, the closed convergence topology is the smallest topology for which the set $\{ (x,y,{\succeq}) : x \succ y\}$ is open in the product topology; see Theorem 3.1 of \citet{kannai1970continuity}.} 

This topology also has the property that, if the experimenter learns that the subject prefers (at least weakly) $x$ to $y$ through her experiments---because the subject chooses $x$ when presented with the pair $\{x,y\}$---this is also reflected in the limiting preference, when it exists: if for some $N \in \Na$ and $x, y \in X$, we have $x \succeq_n y$ for all $ n \ge N$, and if ${\succeq}_n \rightarrow {\succeq}^*$, then $x \succeq^* y$. And our model allows for the possibility that certain parts of the subject's preference remain unobserved. For example, the experimenter cannot learn, from a finite number of observations, that $x$ is strictly preferred to $y$, although she may learn that $x$ is weakly preferred. In this case, if we can still ensure a unique limiting preference ${\succeq}^*$, then it means that correct inferences about missing observations have been made. The closed convergence topology is therefore well suited to the concept of weak rationalization.

\subsection{On the connection with econometrics.}\label{subsec:On the connection with econometrics}

Theorem~\ref{thm:2} is an instance of the consistency of M-estimators. Consistency is well known to rely on three properties of the econometric environment (see, for example, Theorem 4.1.1 of \citealp{amemiya1985advanced}, or Theorems 2.1 and~3.1 of \citealp{newey1994large}).
The first is that the true parameter has to be a unique extremum of the population version of the objective function. 
We prove this property in Lemma~\ref{lem:uniquemax}. The second is the uniform convergence of the sample objective function to the population version. We do not need to assume this property in an ad-hoc fashion; instead, we are able to derive it from our model primitives. Finally, the canonical results on M-estimators need that the parameter space is compact. The topology we use on the space of preferences guarantees its compactness. So, even though our estimation problem is fully nonparametric, we are able to work with a compact space of parameters and we do not need to assume uniform convergence or stochastic equicontinuity.

The most related work in the econometrics literature is \citet*{chernozhukovhongtamer}, who present consistent estimators for partially identified models from moment conditions. We differ from their work in that they provide a general methodology for parametric estimation, while we are specifically interested in the approximation of preferences from pairwise comparisons, and our problem is nonparametric. Because their methodology aims at being general, their main consistency result (Theorem 3.1) assumes the uniform convergence property (part of condition C1). It is not derived from the revealed preference questions that motivate their study. Our consistency result also depends on an analogous uniform convergence property (as mentioned above, this is true generally of the consistency of M-estimators), but it is obtained as a consequence of the primitives of our model. Obtaining consistency directly from the model primitives is the key contribution of Theorem~\ref{thm:2}.  In certain environments, revealed preference conditions are representable by means of moment inequalities, and the results of \citeauthor{chernozhukovhongtamer} can be applied, \citet*{blundell2008best} is a notable instance of such an application in the context of demand analysis. 

\subsection{On parameterized sets of preferences.}\label{subsec:On parameterized sets of preferences}

The metric $\rho$ used in our results can be any metric on the space of preferences that is compatible with the topology of closed convergence. 
For instance, for compact sets of alternatives that satisfy Assumption~\ref{cond:1}, we can use for $\rho$ the Hausdorff distance between subsets of $X \times X$. 

When the space of preferences $\P$ is parameterized, one may wish to work with a metric on preference parameters, such as the usual Euclidean distance when parameters are finite-dimensional vectors.  
Doing so is possible when the space of preferences and the space of parameters are homeomorphic, as formalized by the following lemma.
\begin{lemma}
    \label{lem:compatible-metric}
    Let $\Theta$ be a compact space of parameters endowed with a metric $D$, and $\phi$ the mapping that sends every parameter in $\Theta$ to a preference in $\P$. Suppose that $\phi$ is one-to-one and onto, and is continuous. If $\rho({\succeq}, {\succeq}')$ is defined as $D(\theta, \theta')$, with $\theta, \theta'$ the parameters respectively associated with ${\succeq}, {\succeq}'$, then $\rho$ is a compatible metric on $\P$.
\end{lemma}
Lemma~\ref{lem:compatible-metric} follows directly from the observation that, under the above conditions, the inverse of $\phi$ is continuous by compactness of $\Theta$.\footnote{Hence, any set of preferences that is open in the topology of closed convergence contains an open ball under the metric $\rho$, and conversely, any open ball of preferences under the metric $\rho$ contains an open set in the topology of closed convergence. Therefore, the topology of closed convergence and the metric topology induced by $\rho$ are identical.}



To illustrate this result in a concrete case, let us return to the expected-utility preference environment of Section~\ref{subsec:Expected utility preferences}. Recall that the set of alternatives $X$ is the set of lotteries over a finite collection of prizes $\{\pi_1, \ldots, \pi_d\}$, represented as the $(d-1)$-dimensional simplex $\Delta^{d-1}$, and $\P$ is the set of nonconstant expected utility preferences.
Each preference of $\P$ is naturally parameterized by its normalized vector of utilities: using the notation in Lemma~\ref{lem:compatible-metric}, let the space of parameters be 
\begin{equation*}
    \Theta = \pc{v \in \Re^d : \sum_{i=1}^d v_i=0, \, \|v\|=1},
\end{equation*}
and $D$ be the Euclidean distance.

Let $\phi$ the mapping that, to each $v \in \Theta$, associates the expected utility preference relation $\succeq$ defined as $p \succeq p'$ if and only if $v \cdot p \ge v \cdot p'$. Let us write $\phi(v) = \Phi(U_v)$, where $U_v$ is the utility function on $X$ defined by $U_v(p) = v \cdot p$, and, as in Section~\ref{sec:Preferences from utilities}, $\Phi(u)$ denotes the preference induced by utility function $u: X \mapsto \Re$.
First, we observe that $v \mapsto U_v$ is continuous, when, as in Section~\ref{sec:Preferences from utilities}, the space of utility functions is endowed with the topology of compact convergence. Second, we observe that $\Phi$ is continuous by Theorem 8 of \citet{border1994dynamic} (see Appendix~\ref{app:proof:utility-revealed}), since by Corollary~\ref{cor:expected-utility-revealed}, $\Phi(u)$ is locally strict when $u$ is defined as $u(p) = v \cdot p$ for $v \in \Theta$. Therefore Lemma~\ref{lem:compatible-metric} applies: if $\rho({\succeq}, {\succeq}')$ is defined as the Euclidean distance between the utility indices of ${\succeq}$ and ${\succeq}'$ respectively, then $\rho$ is a compatible metric on $\P$.

We use this fact in Proposition~\ref{cor:expected-utility-stat} to derive convergence rates of estimated preferences in terms of the Euclidean distance on utility indexes.

\subsection{On the transitivity of preferences.}\label{subsec:On the transitivity of preferences}

Aside from the assumption of local strictness, our framework applies to very general classes of preferences. In particular, it applies to preferences without classical rationality hypotheses, such as transitivity.  Still, one may wish to look for rationalizing preferences that are transitive. While it is perfectly reasonable to focus on transitive preferences, one must interpret Theorem~\ref{thm:1} with care. Even when all the preferences that rationalize the observed behavior for a set of experiments can be chosen to be transitive, there is no guarantee that the limiting preference is transitive, even if Assumption~\ref{cond:1} is satisfied. This fact owes to an example of \citet{grodal1974note}.

Adapted to our context, Grodal's example proceeds as follows. Figure~\ref{fig:nontrans} exhibits a nontransitive relation borrowed from \citet*{grodal1974note}, with $X=\Re_{++}^2$ (say, $X$ is a commodity space with two goods). The lines depict indifference curves. All the green indifference curves intersect at one point: $(1/2,1/2)$. Aside from the point $(1/2,1/2)$, $x$ is at least as good as $y$ if and only if it lies on a (weakly) higher indifference curve.  But, all bundles on an indifference curve passing through $(1/2,1/2)$ are indifferent to $(1/2,1/2)$.  This feature makes the preference nontransitive; specifically, the indifference part of the preference is intransitive here. Let $\succeq^*$ denote this preference.

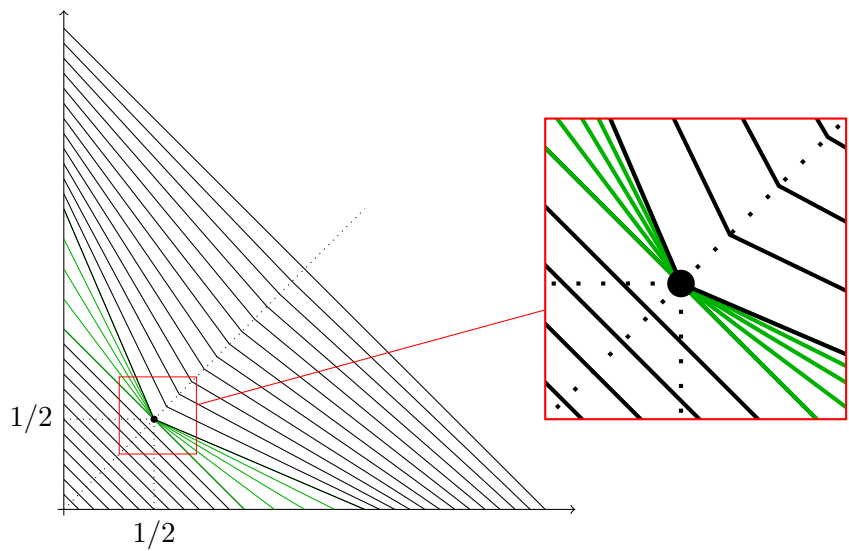
\begin{figure}
\begin{center}
\begin{tikzpicture}[scale=.8, spy using outlines={magnification=3.9, size=4cm, connect spies}]
	\draw[->] (0,-.1) -- (0,8.3) node[anchor=west] {};
	\draw[->] (-.1,0) -- (8.5,0)  node[anchor=north] {};
	\draw[dotted] (0,0) -- (5,5);
	\foreach \i in {0,.25,...,3} {\draw[-,thin] (\i,0) -- (0,\i);}
	\foreach \i in {3,3.5,...,5} {\draw[-,thin,green!70!black] (\i,0) -- (1.5, 1.5) -- (0,\i);}
	\foreach \i in {5,5.25,...,8} {\draw[-,thin] (\i,0) -- (-8/3 + 2.5*\i/3,-8/3 + 2.5*\i/3) -- (0, \i);}
	\draw [fill] (1.5,1.5) circle (.05cm);
	\draw [thin, dotted] (0,1.5) node[anchor=east] {\small $1/2$}-- (1.5,1.5);
	\draw [thin, dotted] (1.5,0) node[anchor=north] {\small $1/2$} -- (1.5,1.5);
	\spy [red] on (1.25,1.25) in node [right] at (8,4);
\end{tikzpicture}\end{center}
\caption{A non-transitive preference}\label{fig:nontrans}
\end{figure}

Imagine an ordered collection of binary choice problems that do not include the alternative $(1/2,1/2)$. Suppose that this collection is either finite or infinite but countable, as the set used in our growing experiments. 
Then for every $n$ there is a ball around $(1/2,1/2)$ that does not contain any alternative in the first $n$ binary choice problems. Consider the preferences pictured in the diagram of Figure~\ref{fig:nontransball}. Compared to the relation depicted in Figure~\ref{fig:nontrans}, the preferences of Figure~\ref{fig:nontransball} have been modified close to $(1/2,1/2)$ so that transitivity holds. Thus, one can construct a sequence of strictly monotone preferences, $\succeq_n$, $n \in \Na$, where each $\succeq_n$ is transitive, and ${\succeq}_n \rightarrow {\succeq}^*$, but $\succeq^*$ is not transitive.

\begin{figure}
\begin{center}
\begin{tikzpicture}[scale=.8, spy using outlines={magnification=3.9, size=4cm, connect spies}]
	\draw[->] (0,-.1) -- (0,8.3) node[anchor=west] {};
	\draw[->] (-.1,0) -- (8.5,0)  node[anchor=north] {};
	\draw[dotted] (0,0) -- (5,5);
	\draw[-,green!70!black,name path=m1] (5,0) -- (1.5, 1.5) ;
	\draw[-,green!70!black,name path=m2] (1.5, 1.5) -- (0,5);
	\draw[-,green!70!black,name path=n1] (3.8,0) -- (1.5, 1.5) ;
	\draw[-,green!70!black,name path=n2] (1.5, 1.5) -- (0,3.8);
	\draw [name path=arcone] (1.8,1.5) arc (0:180:.3cm);
	\draw [name path=arctwo] (1.8,1.5) arc (360:180:.3cm);
	\draw [fill, white] (1.5,1.5) circle (.3cm);
	\foreach \i in {0,.25,...,3} {\draw[-,thin] (\i,0) -- (0,\i);}
	\foreach \i in {5.25,5.5,...,8} {\draw[-,thin] (\i,0) -- (-8/3 + 2.5*\i/3,-8/3 + 2.5*\i/3) -- (0, \i);}
	\draw [blue,thin,fill, fill opacity=.2]  (1.5,1.5) circle (.3cm);
	\path [name intersections={of=n1 and arctwo,by=l1}];
	\path [name intersections={of=n2 and arcone,by=l2}];
	\path [name intersections={of=m1 and arctwo,by=o1}];
	\path [name intersections={of=m2 and arcone,by=o2}];
	\draw [green!70!black] (l1) -- (l2);\draw [green!70!black] (o1) -- (o2);
	\draw [fill] (1.5,1.5) circle (.05cm); 
	\spy [red] on (1.25,1.25) in node [right] at (8,4);
\end{tikzpicture}
\end{center}
\caption{A transitive preference}\label{fig:nontransball}
\end{figure}
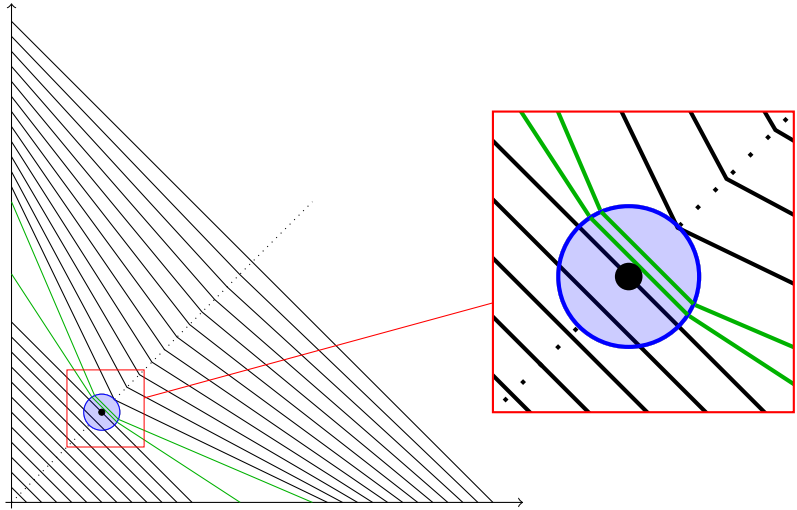

It is generally true that if each $\succ_n$ (the strict part of $\succeq_n$) is transitive, then $\succ^*$ will be transitive as well (see \citealp{grodal1974note}), but in some cases one may desire full transitivity of $\succeq^*$.\footnote{Relations for which the strict part is transitive are usually called quasitransitive, and they possess many of the useful properties possessed by transitive relations. For example, continuous quasitransitive relations possess maximums on compact sets \citep{bergstrom1975}.} The central element of the example above is that the indifference curves get ``squeezed'' together too rapidly. 

There are several ways out of Grodal's example, if we wish to obtain a transitive limiting preference $\succeq^*$. We have seen a few in Sections~\ref{sec:Preferences from utilities} and~\ref{sec:Application to monotone preferences}. A rather general  approach involves Lipschitz conditions on the class of preferences being considered.  

Let us apply the Lipschitz approach to the environment described in Section~\ref{sec:Preferences from utilities}. 
Fix $X=\Re_+^d$ as the set of alternatives, such as a commodity space with $d$ goods. Also fix $a,b\in\Re_{++}$ with $a < b$, and consider the class of utility functions $\U$ defined as the set of all the continuous utility functions $u: \Re_+^d \rightarrow \Re$ such that, for all $i$, and all $x_i, y_i \in \Re_+$ with $x_i < y_i$,
\begin{equation*}
	a \cdot (y_i-x_i)\leq u(y_i,x_{-i})- u(x_i,x_{-i})\leq b \cdot (y_i-x_i)
\end{equation*}
for all $x_{-i} \in \Re_+^{d-1}$.
Hence, each utility function in $\U$ is Lipschitz-bounded above and below. Clearly, every such utility function also describes a transitive, locally strict preference, because $a$ is positive. And by the Arzela-Ascoli Theorem (Theorem 6.4, p. 267, of \citealp{dugundji}), $\U$ is compact. Therefore one can appeal to Proposition~\ref{prop:utility-revealed}, and if each rationalizing preference $\succeq_n$ of the $n$-th experiment of an exhaustive sequence is included in $\Phi(\U)$, then the limiting preference exists and is a member of $\Phi(\U)$, and so is transitive.

\pagebreak

\appendix

\section{Proof of Lemma \ref{lem:1}}\label{app:proof:lem-1}

Suppose, by means of contradiction, that there exist $x,y \in X$ for which $x \succeq_A y$ and $y \succ_B x$.  By continuity of $\succeq_B$ and local strictness of $\succeq_A$, we can assume without loss that $x \succ_A y$ and $y \succ_B x$. Then, by continuity of $\succeq_A$ and $\succeq_B$, and by denseness of the collection $\{x_k, y_k : k \in \Na\}$, there exists $a,b \in \{x_k, y_k : k \in \Na\}$ such that $a \succ_A b$ and $b \succ_B a$.  However, by completeness of ${\succeq}$, either $a \succeq b$, contradicting the implication $a \succeq b \implies a \succeq_B b$, or $b \succeq a$, contradicting the implication $b \succeq a \implies b \succeq_A a$.

\section{Proof of Theorem \ref{thm:1}}\label{app:proof:thm-1}

The proof utilizes Lemma~\ref{lem:1} and the following two elementary lemmas.

\begin{lemma}
	\label{lem:inclusion}
If $A \subseteq X\times X$, then $\{{\succeq} \in X \times X : A \subseteq {\succeq} \}$ is closed.
\end{lemma}
\begin{proof}  
	Let $\{ {\succeq}_n \}_{n \in \Na}$ be a sequence in the set $\{{\succeq} \in X \times X  : A \subseteq {\succeq} \}$ such that ${\succeq}_n \rightarrow {\succeq}$.  Then for all $(x,y) \in A$, we have $x \succeq_n y$, hence $x \succeq y$.  So $(x,y)\in {\succeq}$.
\end{proof}

\begin{lemma}\label{lem:compact}
    The set of all continuous binary relations on $X$ is a compact metrizable space.
\end{lemma}
\begin{proof}
	See Theorem 2 in Chapter B of \citet{hildenbrand2015core}, or Corollary 3.95 of \citet{aliprantis2006infinite}.
\end{proof}

We now return to the main proof of Theorem \ref{thm:1}. 
By assumption, $\P$ is closed, and hence compact as a closed subset of a compact space by Lemma~\ref{lem:compact}. 

Let $\succeq'$ be any complete binary relation (not necessarily in $\P$) such that for all $n$ and all $\{x,y\} \in \Sigma_n$, $x=c(\{x,y\})$ if and only if $x \succeq' y$ ($\succeq'$ is guaranteed to exist because the experiments are nested, $\Sigma_n\subseteq \Sigma_{n+1}$ for all $n$). Similarly, let $\succeq_n'$ be the revealed preference relation that captures the observations made on the experiment $\Sigma_n$, that is, $x \succeq_n' y$ if and only if there is $\{x,y\} \in \Sigma_n$ with $x = c({x,y})$.

For every $n \in \Na$, let $P_n = \{{\succeq} \in \P : {\succeq}'_n \subseteq {\succeq} \}$ be the set of relations in $\P$ that rationalize the observed choices on $\Sigma_n$. 
Lemma~\ref{lem:inclusion} implies that $P_n$ is closed, and hence
compact. Thus, $\{ P_n \}_{n \in \Na}$ constitutes a decreasing sequence of closed sets lying in the compact set $P$, and by the finite intersection property, $\bigcap_{n \in \Na} P_n \neq \varnothing$.  

So let ${\succeq}^* \in \bigcap_{n \in \Na} P_n$.  We claim that $\bigcap_{n \in \Na} P_n = \{ {\succeq}^* \}$. Take any ${\succeq} \in \bigcap_{n \in \Na} P_n$. By definition, for any binary choice problem $\{x,y\} \in \bigcup_{n \in \Na} \Sigma_n$, if $x \succeq' y$ then $x \succeq^* y$ and $x \succeq y$. Hence Lemma~\ref{lem:1} implies ${\succeq} = {\succeq}^*$.

The result now follows as for each $n \in \Na$, $P_n$ is compact, and $\bigcap_{n \in \Na} P_n =\{{\succeq}^*\}$.

\section{Proof of Theorem \ref{thm:2}}\label{app:proof:thm-2}

Throughout this proof, $\succeq^*$ is the subject's preference, $\succeq_n$ is the Kemeny-minimizing estimator for the $n$-th experiment $\Sigma_n$, and $R_n$ is the revealed preference relation for that $n$-th experiment as described in Section~\ref{sec:Main results}. Since $\succeq^*$ is fixed throughout, we write $\mu({\succeq},{\succeq}^*)$ as $\mu({\succeq})$.

To simplify notation, we also let $\bar d_n({\succeq}, R_n) = \frac{1}{n} \abs{{\succeq}\cap R_n}$, noting that  
\begin{equation*}
	d_n ({\succeq},R_n) = \frac{1}{n} \abs{R_n \setminus {\succeq}} = 1 - \bar d_n({\succeq}, R_n). 	
\end{equation*}
In particular, $\succeq_n$ maximizes ${\succeq} \mapsto \bar d_n({\succeq}, R_n)$.

The proof makes use of the following three lemmas.

\begin{lemma}\label{lem:uniquemax}
    For any preference ${\succeq}$ in the class $\P$, if ${\succeq}$ and ${\succeq}^*$ are distinct, then $\mu({\succeq}^*)>\mu({\succeq})$.
\end{lemma}
\begin{proof}
    First, we observe that ${\succeq} \neq{\succeq}^*$ implies that ${\succ}^* \setminus {\succ}$ contains a non-empty open set.
    Indeed, suppose first that $x,y \in X$ satisfies $x \succeq y$ but $(x,y) \notin {\succeq}^*$. 
    By completeness, $y \succ^* x$.
    Then, by continuity, there are neighborhoods $U$ and $V$ of $x$ and $y$ respectively with $V \succ^* U$ (that is, every alternative in $V$ is ranked strictly above every alternative in $U$ according to $\succeq^*$).
    Because $x \succeq y$, by local strictness, there exists $(x',y')\in U \times V$ with $x'\succ y'$. By continuity again there are neighborhoods $U'\ni x'$ and $V'\ni y'$ with $U'\succ V'$. Thus $(y',x')\in (V\cap V')\times (U\cap U')\subseteq {\succ}^*\setminus  {\succ}$. Second, consider the case where there is  $x,y \in X$ with $x \succeq^* y$ but $(x,y) \notin {\succeq}$. This means that $y\succ x$, so there are  neighborhoods $U\ni x$ and $V\ni y$ with $V\succ X$. By local strictness, there is $(x',y')\in U\times V$ with $x'\succ^* y'$. Thus $(y',x')\in{\succeq}\setminus{\succeq}^*$ and we are in the situation above.
    
    Let $q(x,y)$ be a short notation for $q(\succeq^*;x,y)$, and for a binary relation $R$, let $\one_R(x,y)$ if and only if $(x,y) \in R$. We show that $\mu({\succeq}^*)- \mu({\succeq}) > 0$ by the sequence of inequalities below:
    \begin{align*} 
    \mu({\succeq}^*)- \mu({\succeq}) & = \int_{X\times X} \big[\one_{\succeq^*}(x,y) q(x,y) - \one_{\succeq}   (x,y) q(x,y)\big] \de\la(x,y) \\
    &=\int_{X\times X} \big[\one_{{\succeq}^*\setminus {\succeq}}(x,y) q(x,y) - \one_{{\succeq}\setminus {\succeq}^*}(x,y)q(x,y) \big] \de\la(x,y) \\
    &=\int_{X\times X} \big[\one_{\succeq^*\setminus \succeq}(x,y) q(x,y) - \one_{\succeq\setminus \succeq^*}(y,x) q(y,x) \big] \de\la(x)\de\la(y)\\
    &=\int_{X\times X} \big[\one_{\succeq^*\setminus \succeq}(x,y) q(x,y) - \one_{\succ^*\setminus \succ}(x,y) q(y,x) \big] \de\la(x)\de\la(y) \\
    &= \int_{X\times X} \one_{\succ^*\setminus \succ}(x,y) \big[q(x,y) - q(y,x)\big] \de\la(x)\de\la(y) \\
    & > 0.
    \end{align*} 
    The third equality obtains with a change of variables. The fourth equality uses the completeness of $\succeq^*$ and $\succeq$, which means that $(y,x)\in\succeq\setminus\succeq^*$ if and only if $x\succ^* y$ and $x\not\succ y$. The fifth equality follows as $\one_{\succeq^*\setminus \succeq}(x,y) $ and 
    $\one_{\succ^*\setminus \succ}(x,y)$ are equal $\la$-almost surely. The final inequality owes to the fact that $q(x,y)>1/2>q(y,x)$ if $(x,y) \in {\succ}^*$, and that there is an open set on which $\one_{\succ^*\setminus \succ}=1$.
\end{proof}

\begin{lemma}\label{lem:continuityofmu}
    The mapping ${\succeq} \mapsto \mu({\succeq})$ is continuous on $\P$.
\end{lemma}
\begin{proof}
    Let ${\succeq} \in \P$ and $\{ {\succsim}_i \}_{i \in \Na}$ be a sequence of preferences in $\P$ with ${\succsim}_i \rightarrow {\succeq}$. If $x \succsim_i y$ for infinitely many values of $i$, then $x \succeq y$. Hence, by completeness of the preference $\succeq$, if the sequence of binary numbers $\{ \one_{(x,y) \in {\succsim}_i} \}_{i \in \Na}$ diverges, then $x \sim y$.
    
    Next, suppose $x \nsim y$. Then, there are two possibilities. Either $\one_{(x,y) \in {\succsim}_i} = 1$ for all $i$ large enough, then $x \succ y$. Or $\one_{(x,y) \in {\succsim}_i} = 0$ for all $i$ large enough, then, by completeness, $\one_{(y,x) \in {\succsim}_i} = 1$ for all $i$ large enough, and so $y \succ x$.
    
    Finally, recall that, by Assumption (3), the set $\{ (x,y) : x \sim y \}$ has $\lambda$-probability 0. Hence,     
    \begin{align*}
    \mu({\succeq}) & = \int_{\{(x,y) : x \nsim y \}} \one_{(x,y) \in {\succeq} } \, q({\succeq}^*; x,y) \de \la(x,y)  \\
    & = \int_{\{(x,y) : x \nsim y \}} \lim_{n \to \infty} \one_{(x,y) \in {\succsim}_i} \, q({\succeq}^*; x,y) \de \la(x,y) \\
    & = \lim_{i \to \infty} \int_{\{(x,y) : x\nsim y \}} \one_{(x,y)\in {\succsim}_i} \, q({\succeq}^*; x,y) \de \la(x,y) \\
    & = \lim_{i \to \infty} \int \one_{{\succsim}_i} \de \mu = \lim_{i \to \infty} \mu({\succsim}_i), \\
    \end{align*} 
    where the interchange between the limit and integration operators follows from Lebesgue dominated convergence.
\end{proof}

\begin{lemma}\label{lem:continuityofd}
    For all ${\succeq}^\star \in \P$ and all $n \in \Na$, the mapping ${\succeq} \mapsto \bar d_n({\succeq}, R_n)$ defined on $\P$ is almost surely continuous at ${\succeq}^\star$.
\end{lemma}
\begin{proof}
    Let ${\succeq} \in \P$ and $\{ {\succsim}_i \}_{i \in \Na}$ be a sequence of preferences in $\P$ with ${\succsim}_i \rightarrow {\succeq}$ as $i \to \infty$.
    Suppose that, for all decision problems $\{x_k,y_k\} \in \Sigma_n$, $x_k \nsim  y_k$. By the same argument as in         Lemma~\ref{lem:continuityofmu}, either $x_k \succsim_i y_k$ for $i$ large enough and then $x_k \succ y_k$, or $y_k \succsim_i x_k$ for $i$ large enough and then $y_k \succ x_k$. Hence, $\bar d_n({\succsim}_i, R_n) \to \bar d_n({\succeq}, R_n)$ as $i \to \infty$.
    And by assumption (3), with probability 1, it is the case that for all decision problems $\{x_k,y_k\} \in \Sigma_n$, $x_k \nsim  y_k$, which concludes the proof.
\end{proof}

Recall that, by Assumption (2), $\P$ is closed, and so compact as a closed subset of the set of all continuous binary relations on $X$, which itself is compact by Lemma~\ref{lem:compact} in Appendix B.
That $\P$ is compact and the continuity statement in Lemma~\ref{lem:continuityofd} imply that $\bar d_n({\succeq}, R_n)$ converges uniformly in probability to $\mu(\succeq)$ over the domain $\P$ by uniform laws of large numbers (for example, Lemma 2.4 of \citet*{newey1994large}). 
Finally, the mapping ${\succeq} \mapsto \mu({\succeq})$ is continuous on $\P$ by Lemma~\ref{lem:continuityofmu}, and is uniquely maximized at $\succeq^*$.
Under the above conditions, standard consistency theorems apply, such as Theorem 2.4 of \citet*{newey1994large}, and thus $\{ {\succeq}_n \}_{n \in \Na}$ converges to $\succeq^*$ in probability.

\section{Proof of Theorem \ref{thm:3}}\label{app:proof:thm-3}

As in the proof of Theorem \ref{thm:2}, throughout this proof, we fix the subject's preference $\succeq^*$, $\succeq_n$ denotes the Kemeny-minimizing estimator for the size-$n$ experiment $\Sigma_n$, and $R_n$ is the revealed preference relation for $\Sigma_n$. We continue to write $\mu({\succeq},{\succeq}^*)$ as $\mu({\succeq})$, and we also let $\bar d_n({\succeq}, R_n) = \frac{1}{n} \abs{{\succeq}\cap R_n}$. Let $V_\P$ be the VC dimension of $\P$, let
\begin{equation*}
	C = 24^2 \frac{V_{\P} \log (4 e^2)}{5n},	
\end{equation*}
and, for each $n$, let
\begin{align*}
	Z_n & = \sup\big\{\bar d_n({\succeq},R_n) - \mu({\succeq}) : {\succeq} \in \P \big\}, \\    
	Y_n & = \sup\big\{\mu({\succeq}) - \bar d_n({\succeq},R_n)  : {\succeq} \in \P \big\}.
\end{align*}

By the bounded differences inequality (for example, Theorem 6.2 of \citealp*{boucheron2013concentration}),
\begin{equation*}
	\Pr\big(Z_n - \E Z_n > t \text{ or } Y_n - \E Y_n > t \big) \leq 2 e^{-2 t^2 n}.	
\end{equation*}
Moreover, by Theorem 13.7 of \citet*{boucheron2013concentration},
\begin{equation*}
	\max \big\{\sqrt{n}\E Z_n , \sqrt{n}\E Y_n \big\} \leq 72 \sqrt{V_{\P} \log (4e^2)},	
\end{equation*}
as long as $n\geq C$.

Hence on the event $\{Z_n - \E Z_n \leq t \text{ and  } Y_n - \E Y_n \leq t\}$, we have 
\begin{align*}
    \mu({\succeq_n}) & \geq \bar d_n({\succeq}_n,R_n) - Z_n \\
        & \geq \bar d_n({\succeq}^*,R_n) - \left(72 \sqrt{\frac{V_{\P} \log (4e^2)}{n}} + t \right) \\
        & \geq \mu({\succeq^*}) - Y_n - \left(72 \sqrt{\frac{V_{\P} \log (4e^2)}{n}} + t\right) \\
        & \geq \mu({\succeq^*}) - 2 \left(72 \sqrt{\frac{V_{\P} \log (4e^2)}{n}} + t \right). 
\end{align*} 
The first inequality follows from the definition of $Z_n$, the second from the definition of the event we are in, and the bound on $\sqrt{n}\E Z_n$, the third from the definition of $\succeq_n$, the fourth from the definition of $Y_n$, and the final inequality from the bound on $\sqrt{n}\E Y_n$.

Thus, as long as 
\begin{equation}\label{eq:boundforr}
    2\left(72 \sqrt{\frac{V_{\P} \log (4e^2)}{n}} + t \right) < r= r(\eta),  
\end{equation} we know that $\rho({\succeq}_n,{\succeq}^*)<\eta$.

Setting $\da = 2 e^{-2 t^2 n}$ so that $t^2 = \ln (2/\da)/ (2 n)$, together with Equation~\eqref{eq:boundforr}, we get 
\begin{equation*}
	144 \sqrt{\frac{V_{\P} \log (4e^2)}{n}} + 2 \sqrt{\frac{\ln (2/\da) }{2 n}}< r. 	
\end{equation*}

Hence,
\begin{equation*}
	N(\eta, \delta) \leq \max\big\{ r^{-2} (144 \sqrt{V_{\P} \log (4e^2)} + \sqrt{2\ln (2/\da)})^2, C \big\},	
\end{equation*}
where the constant $C$ does not depend on $\da$ nor $\eta$, and hence, as $\da \to 0$ and $\eta \to 0$,
\begin{equation*}
	N(\eta, \delta) = O\left(\frac{1}{r(\eta)^2} \ln \frac{1}{\delta} \right).	
\end{equation*}

\section{Proof of Proposition \ref{prop:utility-revealed}}\label{app:proof:utility-revealed}

Proposition~\ref{prop:utility-revealed} follows from the following result by \citet{border1994dynamic}.

\begin{theorem}[Theorem 8 of \citealp{border1994dynamic}]
    \label{thm:border}
    Let $X$ be a locally compact and separable metric space. 
    Let the space of continuous preference relations on $X$ be endowed with the topology of closed convergence. 
    Let the space of continuous functions on $X$ be endowed with the topology of compact convergence.  
    If $\Phi(u)$ is locally strict, then $\Phi$ is continuous at $u$.
\end{theorem}

Under the conditions stated in Proposition~\ref{prop:utility-revealed}, $\Phi$ is continuous on $\U$ by Theorem~\ref{thm:border}, so $\Phi(\U)$ is compact as the continuous image of a compact set, and hence meets Assumption~\ref{cond:2}.

\section{Proof of Corollary \ref{cor:expected-utility-revealed}} \label{app:proof:expected-utility-revealed}

As explained in Section~\ref{subsec:Expected utility preferences}, any nonconstant expected utility preference can be represented by a member of $\V \equiv \{u \in \Re^d : \|u\|=1 \mbox{ and } \sum_i u_i = 0\}$, which is a compact set in the Euclidean topology (recall that $\| \cdot \|$ refers to the Euclidean norm).

First, observe that the set of functions $\{ U_v : v \in \V \}$, where $U_v$ is defined as $U_v(p) = v \cdot p$, is sequentially compact in the topology of compact convergence. Indeed, let $\{ U_{v^n} \}_{n \in \Na}$ be a sequence of functions where $v^n \in  \V$, and let $\{ v^{n_k} \}_{k \in \Na}$ be a convergent subsequence that converges to $v^*$.  Then $|U_{v^{n_k}}(p) - U_{v^*}(p)|=|(v^{n_k}-v^*)\cdot p|\leq \sqrt{\|v^{n_k}-v^*\|\|p\|}\leq \sqrt{\|v^{n_k}-v^*\|}$, where the first inequality is a Cauchy-Schwarz inequality. Because $\Delta^{d-1}$ is compact, in the present case the topology of compact convergence is a metric topology, and the notions of compactness and sequential compactness are equivalent, so $\{ U_v : v \in \V \}$ is compact.

Second, each nonconstant expected utility preference is locally strict. Indeed, take $p, p' \in \Delta^{d-1}$ and suppose that $v \cdot p \geq v \cdot p'$. Let $p_*,p'_* \in \Delta^{d-1}$ for which $v \cdot p_*> v \cdot p_*'$ (such a pair exists because $v$ represents a nonconstant preference). 
Then, for any $\alpha \in (0,1)$,  $v \cdot (\alpha p_* + (1-\alpha)p) > v \cdot (\alpha p_*' + (1-\alpha)p')$. Local strictness follows by choosing $\alpha$ arbitrarily small.

Corollary~\ref{cor:expected-utility-revealed} then follows from applying Proposition~\ref{prop:utility-revealed}.

\section{Proof of Corollary~\ref{cor:expected-utility-stat}}\label{app:proof:expected-utility-stat}

The proof proceeds by computing an asymptotic lower bound on $r(\eta)$ defined in Equation~\eqref{eq:eta-definition} in Section~\ref{subsec:Convergence in statistical preference models} and then applying Theorem~\ref{thm:3}. 

Recall that $X$ is the simplex $\Delta^{d-1}$, and $\P$ is the set of nonconstant expected utility preferences. It will be convenient to refer to the elements of the simplex by the generic symbols for alternatives $x$ and $y$, as opposed to $p$ and $p'$.
Throughout, each preference of $\P$ is associated to a unique normalized vector of utility indices in the set $\{v \in \Re^d : \sum_{i=1}^d v_i=0, \, \|v\|=1\}$.
For $x \in \Delta^{d-1}$ and $\ep > 0$, we let $\B_\ep(x)$ be the open ball of radius $\ep$ and center $x$ in $\{ z \in \Re^d : \sum_{i=1}^d z_i = 1 \}$, which is the affine span of the simplex $\Delta^{d-1}$.

Let us start with the following lemma. 

\begin{lemma}
	\label{lem:two-separate-balls-1}
	Let $0 < \eta <1$ and $\succeq_u, \succeq_v \in \P$ with $\rho({\succeq}_u, {\succeq}_v) \ge \eta$.
	There exists $x_0, y_0 \in X$ such that, for all $x \in \B_{\eta'}(x_0)$ and $y \in \B_{\eta'}(y_0)$,
	\begin{align*}
		u \cdot x &\ge u \cdot y + \frac{\eta^2}{4 d}, \\
		v \cdot y &\ge v \cdot x + \frac{\eta^2}{4 d}.		
	\end{align*}
	where $u$ and $v$ are the normalized utility indices associated respectively with $\succeq_u$ and $\succeq_v$, and $\eta' \equiv \eta^2 / (10 d)$. In addition, $\B_{\eta'}(x_0) \times \B_{\eta'}(y_0) \subset X \times X$.
\end{lemma}
\begin{proof}
	Let
	\begin{equation*}
		x_0 = \frac{1}{d} \one + \pp{\frac{1}{d} - \eta'} u, \quad \text{ and } \quad y_0 = \frac{1}{d} \one + \pp{\frac{1}{d} - \eta'} v.
	\end{equation*}
	The inequality $\rho({\succeq}_u, {\succeq}_v) \ge \eta$ is equivalent to $\|u - v\| \ge \eta$, or $u \cdot v \le 1 - \eta^2/2$.
	Let $x \in \B_{\eta'}(x_0)$ and $y \in \B_{\eta'}(y_0)$. 
	The following sequence of inequalities obtains:
	\begin{align*}
		u \cdot x 
		&= u \cdot (x-x_0) + u \cdot x_0 \\
		&\ge u \cdot x_0 - \eta' \\
		&\ge u \cdot y_0 + \pp{\frac{1}{d} - \eta'} \frac{\eta^2}{2} - \eta' \\
		&= u \cdot (y_0-y) + u \cdot y + \pp{\frac{1}{d} - \eta'} \frac{\eta^2}{2} - \eta' \\
		&\ge u \cdot y + \pp{\frac{1}{d} - \eta'} \frac{\eta^2}{2} - 2 \eta' \\
		&\ge u \cdot y + \frac{\eta^2}{4 d}.
	\end{align*}
	The first inequality owes to the fact that $| u \cdot (x-x_0) | \le \| u \| \| x-x_0 \| \le \eta'$. Similarly, $| u \cdot (y-y_0) | \le \eta'$ which yields the third inequality. The second inequality comes from $u \cdot x_0 = 1/d - \eta'$ and 
	\begin{equation*}
		u \cdot y_0 = \pp{\frac{1}{d} - \eta'} u \cdot v \le  \pp{\frac{1}{d} - \eta'} -  \pp{\frac{1}{d} - \eta'} \frac{\eta^2}{2}.
	\end{equation*}
	The fourth inequality owes to the fact that $\eta < 1$ and
	\begin{equation*}
		\pp{\frac{1}{d} - \eta'} \frac{\eta^2}{2} - 2 \eta' = \frac{\eta^2}{4d} + \frac{\eta^2 - \eta^4}{20 d} \ge \frac{\eta^2}{4d}.
	\end{equation*}
	By a symmetric argument, we also have
	\begin{equation*}
		v \cdot y \ge v \cdot x + \frac{\eta^2}{4 d}.			
	\end{equation*}
	Finally, observe that $\eta'$ is chosen small enough to ensure that both $\B_{\eta'}(x_0)$ and $\B_{\eta'}(y_0)$ are included in the simplex $\Delta^{d-1}$.
\end{proof}

We now return to the main proof. Let us fix the subject's preference $\succeq^*$, and let $\succeq$ be any preference of $\P$ with $\rho({\succeq}^*, {\succeq}) \ge \eta$, with $0 < \eta < 1$.
As in the proofs of Theorems~\ref{thm:2} and~\ref{thm:3}, we use $q(x,y)$ as a short notation for $q(\succeq^*;x,y)$, and for a binary relation $R$, we let $\one_R(x,y) = 1$ if and only if $(x,y) \in R$. 

We established in the proof of Theorem~\ref{thm:2} that
\begin{equation*}
	\mu({\succeq}^*)- \mu({\succeq})
	= 
	\int_{X \times X} \one_{{\succ}^* \setminus {\succ}}(x,y) \big[q(x,y) - q(y,x)\big] \de\la(x, y).
\end{equation*}

By Lemma~\ref{lem:two-separate-balls-1}, there exists $x_0, y_0 \in X$ such that $\B_{\eta'}(x_0) \times \B_{\eta'}(y_0) \subset X \times X$, and if $(x,y) \in \B_{\eta'}(x_0) \times \B_{\eta'}(y_0)$ (with $\eta' \equiv \eta^2 / (10 d)$) then $x \succ^* y$ while $y \succ x$. Also, if $x \succ^* y$, then $q(x,y) - q(y,x) \ge 0$. Hence,
\begin{align*} 
	\mu({\succeq}^*)- \mu({\succeq}) 
	& = \int_{{\succ}^* \setminus {\succ}} \big[q(x,y) - q(y,x)\big] \de\la(x, y) \\    
	& \geq \int_{\B_{\eta'}(x_0) \times \B_{\eta'}(y_0)} \big[q(x,y) - q(y,x)\big] \de\la(x, y) \\
 & \geq \inf \big\{q(x,y) - q(y,x) :  (x,y) \in \B_{\eta'}(x_0) \times \B_{\eta'}(y_0) \big\} 
 \\ & \hphantom{\geq} \times \la\big( \B_{\eta'}(x_0) \big) \times \la\big( \B_{\eta'}(y_0) \big).
\end{align*}

The Lebesgue measure of each of the $(d-1)$-dimensional balls $\B_{\eta'}(x_0)$ and $\B_{\eta'}(y_0)$ is
\begin{equation*}
	\frac{\pi^{(d-1)/2}}{\Gamma\big((d-1)/2+1 \big) } {\eta'}^{d-1},
\end{equation*}
where $\Gamma$ is the Gamma function, so
\begin{equation*}
	\la\big( \B_{\eta'}(x_0) \big) \times \la\big( \B_{\eta'}(y_0) \big) = \Omega(\eta^{4(d-1)})
\end{equation*}
as $\eta \rightarrow 0$, where the big Omega notation refers to the asymptotic lower bound.

Since $x \in \B_{\eta'}(x_0)$ and $y \in \B_{\eta'}(y_0)$ implies $x \succ^* y$, by Equation~\eqref{eq:q-restricted-1}, 
\begin{equation*}
	q(x,y) - q(y,x) \ge 2 C |u \cdot x - u \cdot y|^k,
\end{equation*}
where $u$ is the normalized vector of utility indices associated with $\succeq^*$. By Lemma~\ref{lem:two-separate-balls-1}, we have
\begin{equation*}
	u \cdot x - u \cdot y \ge \frac{\eta^2}{4d}
\end{equation*}
and hence, 
\begin{equation*}
	\inf \big\{q(x,y) - q(y,x) :  (x,y) \in \B_{\eta'}(x_0) \times \B_{\eta'}(y_0) \big\} 
	\geq 2 C \pp{\frac{\eta^2}{4d}}^k
	= \Omega(\eta^{2k})
\end{equation*}
as $\eta \rightarrow 0$.

Overall, we get $\mu({\succeq}^*)- \mu({\succeq}) = \Omega( \eta^{4(d-1) + 2k} )$, and thus $r(\eta) = \Omega( \eta^{4(d-1) + 2k} )$.
Applying Theorem~\ref{thm:3} and observing that the VC dimension of $\P$ is no greater than $d+1$ (and so finite) by Proposition 4.20 of \citet{wainwright2019high}, we have
\begin{equation*}
	N(\eta,\da) = O\pp{ \frac{1}{\eta^{8(d-1) + 4k} } \ln \frac{1}{\da} }.
\end{equation*}

\section{Proof of Lemma \ref{lem:monotone-closed}}\label{app:proof:monotone-closed}

Let $\{\succeq_n\}_{n \in \Na}$ be a converging sequence of Grodal-transitive preferences that are strictly monotone with respect to $\rhd$, and let $\succeq^*$ be the limiting binary relation.

Recall that, by Lemma~\ref{lem:compact} in Appendix~\ref{app:proof:thm-1}, the set of continuous binary relations is closed, and so $\succeq^*$ is continuous.
Also, for each $x,y \in X$, either $x \succeq_n y$ or $y \succeq_n x$, so there is a subsequence $\{\succeq_{n_k} \}_{k \in \Na}$ for which either $x \succeq_{n_k} y$ for all $k$, or for which $y \succeq_{n_k} x$ for all $k$. Hence, either $x \succeq y$ or $y \succeq x$, which makes $\succeq^*$ complete. Hence, $\succeq^*$ is a preference.

Suppose by means of contradiction that $\succeq^*$ is not strictly monotone with respect to $\rhd$. In that case, there are $x,y \in X$ for which $x \rhd y$ and yet $y \succeq^* x$. 
Let $\{x_n\}_{n \in \Na}$, $\{y_n\}_{n \in \Na}$ be any sequences of alternatives in $X$ that converge to $x$ and $y$ and respectively and for which $y_n \succeq_n x_n$ for all $n$ (existence of sequences satisfying this property follows from the definition of closed convergence in Section~\ref{sec:Main results}). Because $\rhd$ is open, for $n$ large enough, $x_n \rhd y_n$, which contradicts the fact that $\succeq_n$ is strictly monotone with respect to $\rhd$. Hence, $\succeq^*$ is strictly monotone.

Finally, we show that $\succeq^*$ is Grodal-transitive.
Suppose $x,y,z,w \in X$ satisfy $x \succeq^* y \succ^* z \succeq^* w$. Let $\{x_n\}_{n \in \Na}$, $\{y_n\}_{n \in \Na}$, $\{z_n\}_{n \in \Na}$, $\{w_n\}_{n \in \Na}$ be sequences of alternatives in $X$ that converge to $x,y,z,w$ respectively, and for which $x_n \succeq_n y_n$ and $z_n \succeq_n w_n$ (which, again, exist by the definition of closed convergence). Since $y \succ^* z$, for $n$ large enough, $y_n \succ_n z_n$. Consequently, for $n$ large, $x_n \succeq_n y_n \succ_n z_n \succeq_n w_n$, which, by Grodal-transitivity, implies $x_n \succeq_n w_n$, and so $x \succeq^* w$.

\section{Proof of Lemma \ref{lem:monotone-locallystrict}}\label{app:proof:monotone-locallystrict}

For any $x \in X$, let $U_x$ be the set $\{y:y\succ x\}$.  
Let us show that for each $x,y \in X$, either $U_x \subseteq U_y$, or $U_y\subseteq U_x$.  
To see this, suppose by means of contradiction that there is $z \in U_x \setminus U_y$ and $w\in U_y\setminus U_x$.  
Then we have $y \succeq z \succ x$ and $x \succeq w \succ y$.  Therefore, $x \succeq w \succ y \succeq z$, which implies $x\succeq z$ by Grodal-transitivity.  This contradicts $z \succ x$.

Now, fix $x, y \in X$ such that $x \succeq y$, and fix a neighborhood $V$ of $(x,y)$ in the product space $X \times X$. By the lemma hypotheses, there exists $(x',y') \in V$ such that $x' \rhd x$ and $y \rhd y'$. Then, there are two possibilities: either either $U_{y'} \subseteq U_{x}$, or $U_x\subseteq U_{y'}$. In the first case, by monotonicity, $y' \succ y$, which implies $y \succ x$, contradicting $x \succeq y$. So, we must have $U_x \subseteq U_{y'}$. Then $x' \in U_x$, so $x' \in U_{y'}$, which implies $x' \succ y'$. Hence, $\succeq$ is locally strict.

\section{Proof of Corollary \ref{cor:preferences-menus-revealed}}\label{app:proof:preferences-menus-revealed}

First, we prove that the set of alternatives $\M$ meets Assumption~\ref{cond:1}.

That $\M$ is locally compact follows from Theorem 1.8.3 of \citet{schneider}, which demonstrates that the set of nonempty, convex subsets of the simplex is compact, and Theorem 18.4 of \citet{willard}. 
The fact that $\M$ is separable obtains from Theorem 3.85(3) of \citet{aliprantis2006infinite}, together with the fact that a subset of a separable metric space is itself separable (Problem 16G Part 1 of \citealp{willard}).
Next, we show that $\M$ is completely metrizable. This space is, by definition, metrizable. 
The Hausdorff metric is complete on the set of nonempty, closed, convex subsets of $\Delta^{d-1}$; let us call this set $\M^*$. 
This fact owes to a straightforward adaptation of Theorem 1.8.2 of \citet{schneider}, together with the fact that $\M^*$ is a Hausdorff-closed set in the space of compact, convex, nonempty subsets of $\{x \in R^d: \sum_i x_i = 1\}$, because a closed subset of a metric space is complete (Theorem 24.10 of \citealp{willard}).  
Further, $\M$ is a Hausdorff (relatively) open subset of $\M^*$. It then follows by Alexandroff's Theorem (Theorem 24.12 of \citealp{willard}) that $\M$ is completely metrizable.

Secondly, we show that the hypothesis of Lemma~\ref{lem:monotone-closed} is satisfied, that is, that the dominance relation $\sqsupset$ is open. 

This result can be obtained by means of a standard isometry between the set $\M$ endowed with the Hausdorff metric, and the set of support functions of members of $\M$ defined on $\U$ (Lemma 8 of \citealp*{dlr}; p. 594 of \citealp*{dlrs}; Theorem 1.8.11 of \citealp*{schneider}).
For a member $m$ of $\M$, such a support function is written $h_m(u) = \sup_{p \in m} u\cdot p$. 
We endow the set of these support functions with the sup-norm metric.\footnote{So that the distance between two functions $f,g$ is given by $\sup_{u \in \U} |f(u)-g(u)|$.}  
Observe that $m_A \sqsupset m_B$ if and only if for every $u\in \U$, $h_{m_A}(u) > h_{m_B}(u)$.  
In particular, since $\U$ is compact, $m_A \sqsupset m_B$ if and only if there is $\ep > 0$ for which $h_{m_A}(u)-h_{m_B}(u) > \ep$.  

For an element $m$ of $\M$, let $\mathcal{B}_\delta(m)$ denote the open ball of radius $\delta$ centered on $m$. Suppose $m_A, m_B \in \M$ satisfy $m_A \sqsupset m_B$. 
Let $m_A' \in \mathcal{B}_{\frac{\ep}{3}}(m_A)$ and  $m_B' \in \mathcal{B}_{\frac{\ep}{3}}(m_B)$.
By the isometry aforementioned, for each $u\in \U$, $h_{A'}(u)>h_A(u)-\frac{\ep}{3}$ and $h_{B'}(u)<h_B(u)+\frac{\ep}{3}$.  
Then $h_{A'}(u)-h_{B'}(u)>h_A(u)-h_B(u)-\frac{2\ep}{3}>\frac{\ep}{3}>0$, so that $m_A' \sqsupset m_B'$. 
Hence $\sqsupset$ is open.

Thirdly, we show that the hypothesis of Lemma~\ref{lem:monotone-locallystrict} is satisfied. Let $m_A \in \M$. For any $\alpha \in (0,1)$, let $m_B$ be the Minkowski sum of $m_A$ and $\Delta^{d-1}$, weighted by $\alpha$ and $1-\alpha$ respectively, i.e., $m_B = \{ \alpha p + (1-\alpha) p' : p \in m_A, p' \in \Delta(X) \}$.
Then, $m_B$ is arbitrarily close to $m_A$ by choosing $\alpha$ small enough, and $m_B \sqsupset m_A$.
Similarly, fix $p_0$ in the interior of $m_A$, and for any $\alpha \in (0,1)$, let $m_C = \alpha p_0 + (1-\alpha) m_A$. Then, $m_C$ is arbitrarily close to $m_A$ by choosing $\alpha$ small enough, and $m_A \sqsupset m_C$.

Hence, $\M$ meets Assumption~\ref{cond:1}, and by Lemmas~\ref{lem:monotone-closed} and~\ref{lem:monotone-locallystrict},  the class $\P$ meets Assumption~\ref{cond:2}.

\footnotesize
\bibliographystyle{ecta}
\bibliography{references}

\end{document}

%% file: macros.tex
    \usepackage[T1]{fontenc}
    \usepackage[utf8]{inputenc}
	\usepackage[english]{babel}
	\usepackage[foot]{amsaddr}
	\usepackage{amssymb,amsmath,amsfonts,amsthm}
	\usepackage{mathtools} 
	\usepackage{enumerate,enumitem}
	\usepackage{graphicx}
	\usepackage{natbib,verbatim,subfigure}
	\usepackage{tikz}
	\usetikzlibrary{decorations,calc,intersections,shapes,spy}

	\usepackage[left=1.25in,top=1.25in,right=1.25in,bottom=1.25in]{geometry}
	\usepackage[onehalfspacing]{setspace}

	\newcommand{\pp}[1]{\left( #1 \right)}
	
	\newcommand{\pc}[1]{\left\{ #1 \right\}}
	
	\newcommand{\g}{\ifnum\currentgrouptype=16 \;\middle|\;\else\mid\fi}
	\newcommand{\df}[1]{\textit{#1}}
	\newcommand{\norm}[1]{\| #1 \|}
	
	\newcommand{\abs}[1]{ \left | #1 \right | }
	\DeclareMathOperator*{\argmin}{{arg\,min}}

	\def\de{\mathop{}\!\mathrm{d}}
	\def\E{\mathop{}\mathrm{E}}
	\def\Pr{\mathop{}\mathrm{Pr}}

	\def\U{\mathcal{U}}
	\def\V{\mathcal{V}}
	\def\B{\mathcal{B}}
	\def\P{\mathcal{P}}
	
	\def\M{\mathcal{M}}

	\def\ep{\varepsilon}

	\def\one{\mathbf{1}}

	\def\la{\lambda}
	\def\da{\delta}
	\def\Re{\mathbf{R}}
	
	\def\Na{\mathbf{N}}

	\newdimen\slantmathcorr
	\def\oversl#1{
	\setbox0=\hbox{$#1$}
	\slantmathcorr=\wd0
	\hskip 0.2\slantmathcorr \overline{\hbox to 0.8\wd0{%
	\vphantom{\hbox{$#1$}}}}
	\hskip-\wd0\hbox{$#1$}
	}
	\def\undersl#1{
	\setbox0=\hbox{$#1$}
	\slantmathcorr=\wd0
	\underline{\hbox to 0.8\wd0{%
	\vphantom{\hbox{$#1$}}}}
	\hskip-0.8\wd0\hbox{$#1$}
	}

	\theoremstyle{plain}
	\newtheorem{theorem}{Theorem}

	\newtheorem{proposition}{Proposition}
	\newtheorem{lemma}{Lemma}
	
	\newtheorem{corollary}{Corollary}

	\theoremstyle{definition}

	\theoremstyle{remark}

	\newtheorem*{claim*}{Claim}

%% file: main_paper.bbl
\begin{thebibliography}{57}
\newcommand{\enquote}[1]{``#1''}
\expandafter\ifx\csname natexlab\endcsname\relax\def\natexlab#1{#1}\fi

\bibitem[\protect\citeauthoryear{Afriat}{Afriat}{1967}]{afriat}
\textsc{Afriat, S.~N.} (1967): \enquote{The Construction of Utility Functions
  from Expenditure Data,} \emph{International Economic Review}, 8, 67--77.

\bibitem[\protect\citeauthoryear{Ahn, Choi, Gale, and Kariv}{Ahn
  et~al.}{2014}]{ahn2014}
\textsc{Ahn, D.~S., S.~Choi, D.~Gale, and S.~Kariv} (2014): \enquote{Estimating
  Ambiguity Aversion in a Portfolio Choice Experiment,} \emph{Quantitative
  Economics}, 5, 195--223.

\bibitem[\protect\citeauthoryear{Aliprantis and Border}{Aliprantis and
  Border}{2006}]{aliprantis2006infinite}
\textsc{Aliprantis, C.~D. and K.~Border} (2006): \emph{Infinite Dimensional
  Analysis: {A} Hitchhiker's Guide}, Springer, 3rd ed.

\bibitem[\protect\citeauthoryear{Amemiya}{Amemiya}{1985}]{amemiya1985advanced}
\textsc{Amemiya, T.} (1985): \emph{Advanced Econometrics}, Harvard University
  Press.

\bibitem[\protect\citeauthoryear{Andrews and Guggenberger}{Andrews and
  Guggenberger}{2009}]{andrews2009validity}
\textsc{Andrews, D.~W. and P.~Guggenberger} (2009): \enquote{Validity of
  Subsampling and `Plug-in Asymptotic' Inference for Parameters Defined by
  Moment Inequalities,} \emph{Econometric Theory}, 25, 669--709.

\bibitem[\protect\citeauthoryear{Basu}{Basu}{2019}]{basu2019learnability}
\textsc{Basu, P.} (2019): \enquote{Learnability and Stochastic Choice,} SSRN
  Working Paper No. 3338991.

\bibitem[\protect\citeauthoryear{Basu and Echenique}{Basu and
  Echenique}{2018}]{basu2018learnability}
\textsc{Basu, P. and F.~Echenique} (2018): \enquote{Learnability and Models of
  Decision Making under Uncertainty,} California Institute of Technology
  Working Paper.

\bibitem[\protect\citeauthoryear{Bergstrom}{Bergstrom}{1975}]{bergstrom1975}
\textsc{Bergstrom, T.~C.} (1975): \enquote{Maximal Elements of Acyclic
  Relations on Compact Sets,} \emph{Journal of Economic Theory}, 10, 403--404.

\bibitem[\protect\citeauthoryear{Blundell, Browning, and Crawford}{Blundell
  et~al.}{2008}]{blundell2008best}
\textsc{Blundell, R., M.~Browning, and I.~Crawford} (2008): \enquote{Best
  Nonparametric Bounds on Demand Responses,} \emph{Econometrica}, 76,
  1227--1262.

\bibitem[\protect\citeauthoryear{Blundell, Kristensen, and Matzkin}{Blundell
  et~al.}{2010}]{blundell2010stochastic}
\textsc{Blundell, R., D.~Kristensen, and R.~L. Matzkin} (2010):
  \enquote{Stochastic Demand and Revealed Preference,} Working paper.

\bibitem[\protect\citeauthoryear{Border and Segal}{Border and
  Segal}{1994}]{border1994dynamic}
\textsc{Border, K.~C. and U.~Segal} (1994): \enquote{Dynamic Consistency
  Implies Approximately Expected Utility Preferences,} \emph{Journal of
  Economic Theory}, 63, 170--188.

\bibitem[\protect\citeauthoryear{Boucheron, Lugosi, and Massart}{Boucheron
  et~al.}{2013}]{boucheron2013concentration}
\textsc{Boucheron, S., G.~Lugosi, and P.~Massart} (2013): \emph{Concentration
  Inequalities: {A} Nonasymptotic Theory of Independence}, Oxford University
  Press.

\bibitem[\protect\citeauthoryear{Brown and Matzkin}{Brown and
  Matzkin}{1996}]{brown1996testable}
\textsc{Brown, D.~J. and R.~L. Matzkin} (1996): \enquote{Testable Restrictions
  on the Equilibrium Manifold,} \emph{Econometrica}, 64, 1249--1262.

\bibitem[\protect\citeauthoryear{Carvalho, Meier, and Wang}{Carvalho
  et~al.}{2016}]{carvalho2016poverty}
\textsc{Carvalho, L., S.~Meier, and S.~W. Wang} (2016): \enquote{Poverty and
  Economic Decision Making: {E}vidence from Changes in Financial Resources at
  Payday,} \emph{The American Economic Review}, 106, 260--284.

\bibitem[\protect\citeauthoryear{Carvalho and Silverman}{Carvalho and
  Silverman}{2019}]{carvalho2017complexity}
\textsc{Carvalho, L. and D.~Silverman} (2019): \enquote{Complexity and
  Sophistication,} NBER working paper No. 26036.

\bibitem[\protect\citeauthoryear{Chambers, Echenique, and Shmaya}{Chambers
  et~al.}{2014}]{CES}
\textsc{Chambers, C.~P., F.~Echenique, and E.~Shmaya} (2014): \enquote{The
  Axiomatic Structure of Empirical Content,} \emph{The American Economic
  Review}, 104, 2303--19.

\bibitem[\protect\citeauthoryear{Chapman, Dean, Ortoleva, Snowberg, and
  Camerer}{Chapman et~al.}{2018}]{chapman2018econographics}
\textsc{Chapman, J., M.~Dean, P.~Ortoleva, E.~Snowberg, and C.~Camerer} (2018):
  \enquote{Econographics,} NBER working paper No. w24931.

\bibitem[\protect\citeauthoryear{Chavas and Cox}{Chavas and
  Cox}{1993}]{chavas1993generalized}
\textsc{Chavas, J.-P. and T.~L. Cox} (1993): \enquote{On Generalized Revealed
  Preference Analysis,} \emph{The Quarterly Journal of Economics}, 108,
  493--506.

\bibitem[\protect\citeauthoryear{Cherchye, De~Rock, and Vermeulen}{Cherchye
  et~al.}{2011}]{cherchye2011revealed}
\textsc{Cherchye, L., B.~De~Rock, and F.~Vermeulen} (2011): \enquote{The
  Revealed Preference Approach to Collective Consumption Behaviour: {T}esting
  and Sharing Rule Recovery,} \emph{The Review of Economic Studies}, 78,
  176--198.

\bibitem[\protect\citeauthoryear{Chernozhukov, Hong, and Tamer}{Chernozhukov
  et~al.}{2007}]{chernozhukovhongtamer}
\textsc{Chernozhukov, V., H.~Hong, and E.~Tamer} (2007): \enquote{Estimation
  and Confidence Regions for Parameter Sets in Econometric Models,}
  \emph{Econometrica}, 75, 1243--1284.

\bibitem[\protect\citeauthoryear{Choi, Kariv, M{\"u}ller, and Silverman}{Choi
  et~al.}{2014}]{choi2014more}
\textsc{Choi, S., S.~Kariv, W.~M{\"u}ller, and D.~Silverman} (2014):
  \enquote{Who Is (More) Rational?} \emph{The American Economic Review}, 104,
  1518--1550.

\bibitem[\protect\citeauthoryear{Clinton, Jackman, and Rivers}{Clinton
  et~al.}{2004}]{clinton_jackman_rivers_2004}
\textsc{Clinton, J., S.~Jackman, and D.~Rivers} (2004): \enquote{The
  Statistical Analysis of Roll Call Data,} \emph{The American Political Science
  Review}, 98, 355–370.

\bibitem[\protect\citeauthoryear{Dekel, Lipman, and Rustichini}{Dekel
  et~al.}{2001}]{dlr}
\textsc{Dekel, E., B.~L. Lipman, and A.~Rustichini} (2001):
  \enquote{Representing Preferences with a Unique Subjective State Space,}
  \emph{Econometrica}, 69, 891--934.

\bibitem[\protect\citeauthoryear{Dekel, Lipman, Rustichini, and Sarver}{Dekel
  et~al.}{2007}]{dlrs}
\textsc{Dekel, E., B.~L. Lipman, A.~Rustichini, and T.~Sarver} (2007):
  \enquote{Representing Preferences with a Unique Subjective State Space: {A}
  Corrigendum,} \emph{Econometrica}, 75, 591--600.

\bibitem[\protect\citeauthoryear{Dugundji}{Dugundji}{1966}]{dugundji}
\textsc{Dugundji, J.} (1966): \emph{Topology}, Bacon.

\bibitem[\protect\citeauthoryear{Falk, Becker, Dohmen, Enke, Huffman, and
  Sunde}{Falk et~al.}{2018}]{falkQJE2018}
\textsc{Falk, A., A.~Becker, T.~Dohmen, B.~Enke, D.~Huffman, and U.~Sunde}
  (2018): \enquote{Global Evidence on Economic Preferences,} \emph{The
  Quarterly Journal of Economics}, 133, 1645--1692.

\bibitem[\protect\citeauthoryear{Forges and Minelli}{Forges and
  Minelli}{2009}]{forges2009}
\textsc{Forges, F. and E.~Minelli} (2009): \enquote{{A}friat's Theorem for
  General Budget Sets,} \emph{Journal of Economic Theory}, 144, 135--145.

\bibitem[\protect\citeauthoryear{Gorno}{Gorno}{2019}]{gorno}
\textsc{Gorno, L.} (2019): \enquote{Revealed Preference and Identification,}
  \emph{Journal of Economic Theory}, 183, 698--739.

\bibitem[\protect\citeauthoryear{Grant, Kline, Meneghel, Quiggin, and
  Tourky}{Grant et~al.}{2016}]{grant2016theory}
\textsc{Grant, S., J.~Kline, I.~Meneghel, J.~Quiggin, and R.~Tourky} (2016):
  \enquote{A Theory of Robust Experiments for Choice under Uncertainty,}
  \emph{Journal of Economic Theory}, 165, 124--151.

\bibitem[\protect\citeauthoryear{Green and Srivastava}{Green and
  Srivastava}{1986}]{green1986expected}
\textsc{Green, R.~C. and S.~Srivastava} (1986): \enquote{Expected Utility
  Maximization and Demand Behavior,} \emph{Journal of Economic Theory}, 38,
  313--323.

\bibitem[\protect\citeauthoryear{Grodal}{Grodal}{1974}]{grodal1974note}
\textsc{Grodal, B.} (1974): \enquote{A Note on the Space of Preference
  Relations,} \emph{Journal of Mathematical Economics}, 1, 279--294.

\bibitem[\protect\citeauthoryear{Halevy, Persitz, and Zrill}{Halevy
  et~al.}{2018}]{halevy2018parametric}
\textsc{Halevy, Y., D.~Persitz, and L.~Zrill} (2018): \enquote{Parametric
  Recoverability of Preferences,} \emph{Journal of Political Economy}, 126,
  1558--1593.

\bibitem[\protect\citeauthoryear{Hildenbrand}{Hildenbrand}{1970}]{HILDENBRAND1970161}
\textsc{Hildenbrand, W.} (1970): \enquote{On Economies with Many Agents,}
  \emph{Journal of Economic Theory}, 2, 161--188.

\bibitem[\protect\citeauthoryear{Hildenbrand}{Hildenbrand}{2015}]{hildenbrand2015core}
---\hspace{-.1pt}---\hspace{-.1pt}--- (2015): \emph{Core and Equilibria of a
  Large Economy}, Princeton University Press.

\bibitem[\protect\citeauthoryear{Jackman}{Jackman}{2001}]{jackman_2001}
\textsc{Jackman, S.} (2001): \enquote{Multidimensional Analysis of Roll Call
  Data via Bayesian Simulation: {I}dentification, Estimation, Inference, and
  Model Checking,} \emph{Political Analysis}, 9, 227–241.

\bibitem[\protect\citeauthoryear{Kannai}{Kannai}{1970}]{kannai1970continuity}
\textsc{Kannai, Y.} (1970): \enquote{Continuity Properties of the Core of a
  Market,} \emph{Econometrica}, 38, 791--815.

\bibitem[\protect\citeauthoryear{Kemeny}{Kemeny}{1959}]{kemeny1959}
\textsc{Kemeny, J.~G.} (1959): \enquote{Mathematics without Numbers,}
  \emph{Daedalus}, 88, 577--591.

\bibitem[\protect\citeauthoryear{Kendall}{Kendall}{1938}]{kendall1938}
\textsc{Kendall, M.~G.} (1938): \enquote{A New Measure of Rank Correlation,}
  \emph{Biometrika}, 30, 81--93.

\bibitem[\protect\citeauthoryear{Kreps}{Kreps}{1979}]{kreps1979representation}
\textsc{Kreps, D.~M.} (1979): \enquote{A Representation Theorem for `Preference
  for Flexibility',} \emph{Econometrica}, 47, 565--577.

\bibitem[\protect\citeauthoryear{K\"ubler, Malhotra, and
  Polemarchakis}{K\"ubler et~al.}{2020}]{polemarchakiskubler2020}
\textsc{K\"ubler, F., R.~Malhotra, and H.~Polemarchakis} (2020):
  \enquote{Identification of Preferences, Demand and Equilibrium with Finite
  Data,} University of Warwick Economics Research Paper No. 1290.

\bibitem[\protect\citeauthoryear{K{\"u}bler and Polemarchakis}{K{\"u}bler and
  Polemarchakis}{2017}]{kubler2015identification}
\textsc{K{\"u}bler, F. and H.~Polemarchakis} (2017): \enquote{The
  Identification of Beliefs from Asset Demand,} \emph{Econometrica}, 85,
  1219--1238.

\bibitem[\protect\citeauthoryear{Manski}{Manski}{2003}]{manski2003partial}
\textsc{Manski, C.~F.} (2003): \emph{Partial Identification of Probability
  Distributions}, Springer.

\bibitem[\protect\citeauthoryear{Mas-Colell}{Mas-Colell}{1974}]{mas1974continuous}
\textsc{Mas-Colell, A.} (1974): \enquote{Continuous and Smooth Consumers:
  {A}pproximation Theorems,} \emph{Journal of Economic Theory}, 8, 305--336.

\bibitem[\protect\citeauthoryear{Mas-Colell}{Mas-Colell}{1978}]{mascolell78}
---\hspace{-.1pt}---\hspace{-.1pt}--- (1978): \enquote{On Revealed Preference
  Analysis,} \emph{The Review of Economic Studies}, 45, 121--131.

\bibitem[\protect\citeauthoryear{Matzkin}{Matzkin}{1991}]{matzkin1991axioms}
\textsc{Matzkin, R.~L.} (1991): \enquote{Axioms of Revealed Preference for
  Nonlinear Choice Sets,} \emph{Econometrica}, 59, 1779--1786.

\bibitem[\protect\citeauthoryear{Matzkin}{Matzkin}{2003}]{matzkin2003nonparametric}
---\hspace{-.1pt}---\hspace{-.1pt}--- (2003): \enquote{Nonparametric Estimation
  of Nonadditive Random Functions,} \emph{Econometrica}, 71, 1339--1375.

\bibitem[\protect\citeauthoryear{Matzkin}{Matzkin}{2006}]{matzkin2006identification}
---\hspace{-.1pt}---\hspace{-.1pt}--- (2006): \enquote{Identification of
  Consumers Preferences When Their Choices Are Unobservable,} in
  \emph{Rationality and Equilibrium: {A} Symposium in Honor of Marcel K.
  Richter}, ed. by C.~D. Aliprantis, R.~L. Matzkin, D.~McFadden, J.~C. Moore,
  and N.~Yannelis, Springer, 195--215.

\bibitem[\protect\citeauthoryear{Matzkin}{Matzkin}{2007}]{matzkin2007heterogeneous}
---\hspace{-.1pt}---\hspace{-.1pt}--- (2007): \enquote{Heterogeneous Choice,}
  in \emph{Advances in Economics and Econometrics: {T}heory and Applications,
  Ninth World Congress of the Econometric Society}, ed. by R.~Blundell,
  W.~Newey, and T.~Persson, Cambridge University Press, vol.~43 of
  \emph{Econometric Society Monographs}, chap.~4, 75--110.

\bibitem[\protect\citeauthoryear{Newey and McFadden}{Newey and
  McFadden}{1994}]{newey1994large}
\textsc{Newey, W.~K. and D.~McFadden} (1994): \enquote{Large Sample Estimation
  and Hypothesis Testing,} in \emph{Handbook of Econometrics}, ed. by R.~Engle
  and D.~McFadden, North Holland, vol.~4, chap.~36, 2111--2245.

\bibitem[\protect\citeauthoryear{Nishimura, Ok, and Quah}{Nishimura
  et~al.}{2017}]{nishimura}
\textsc{Nishimura, H., E.~A. Ok, and J.~K.-H. Quah} (2017): \enquote{A
  Comprehensive Approach to Revealed Preference Theory,} \emph{The American
  Economic Review}, 107, 1239--1263.

\bibitem[\protect\citeauthoryear{Polemarchakis, Selden, and Song}{Polemarchakis
  et~al.}{2017}]{polemarchakis2016identification}
\textsc{Polemarchakis, H., L.~Selden, and X.~Song} (2017): \enquote{The
  Identification of Attitudes Towards Ambiguity and Risk from Asset Demand,}
  Columbia Business School Research Paper No. 17--43.

\bibitem[\protect\citeauthoryear{Poole and Rosenthal}{Poole and
  Rosenthal}{1985}]{poolerosenthal1985}
\textsc{Poole, K.~T. and H.~Rosenthal} (1985): \enquote{A Spatial Model for
  Legislative Roll Call Analysis,} \emph{American Journal of Political
  Science}, 29, 357--384.

\bibitem[\protect\citeauthoryear{Schneider}{Schneider}{1993}]{schneider}
\textsc{Schneider, R.} (1993): \emph{Convex Bodies: the {B}runn-{M}inkowski
  Theory}, Cambridge University Press, 1st ed.

\bibitem[\protect\citeauthoryear{Varian}{Varian}{1982}]{varian1982nonparametric}
\textsc{Varian, H.~R.} (1982): \enquote{The Nonparametric Approach to Demand
  Analysis,} \emph{Econometrica}, 50, 945--973.

\bibitem[\protect\citeauthoryear{von Gaudecker, van Soest, and Wengstrom}{von
  Gaudecker et~al.}{2011}]{vonGaudecker2011}
\textsc{von Gaudecker, H.-M., A.~van Soest, and E.~Wengstrom} (2011):
  \enquote{Heterogeneity in Risky Choice Behavior in a Broad Population,}
  \emph{The American Economic Review}, 101, 664--94.

\bibitem[\protect\citeauthoryear{Wainwright}{Wainwright}{2019}]{wainwright2019high}
\textsc{Wainwright, M.~J.} (2019): \emph{High-Dimensional Statistics: {A}
  Non-Asymptotic Viewpoint}, Cambridge University Press.

\bibitem[\protect\citeauthoryear{Willard}{Willard}{2004}]{willard}
\textsc{Willard, S.} (2004): \emph{General Topology}, Courier Corporation.

\end{thebibliography}
